\documentclass[12pt,a4paper]{article}

\usepackage{amsmath}
\usepackage{amsthm}
\newtheorem*{prop}{Proposition}
\newtheorem*{lemma}{Lemma}
\newtheorem*{theorem}{Theorem}
\numberwithin{equation}{section}

\title{Modified method of simplest equation for obtaining exact analytical solutions of nonlinear partial differential equations: Further development of
methodology with two applications}
\author{Nikolay K. Vitanov$^{1,2,}$\footnote{corresponding author, e-mail: vitanov@imbm.bas.bg}, Zlatinka I. Dimitrova$^3$, Kaloyan N. Vitanov$^1$}
\date{$^1$ Institute of Mechanics, Bulgarian Academy of Sciences, Acad. G. Bonchev Str., Bl. 4,
1113 Sofia, Bulgaria\\
$2$ Max-Plank Institute for the Physics of Complex Systems, N{\"o}thnitzer Str. 38, 01187
Dresden, Germany\\
$3$ ''G. Nadjakov" Institute of Solid State Physics, Tzarigradsko Chaussee Blvd. 72, 1784 Sofia,
Bulgaria}
\begin{document}
\maketitle
\begin{abstract}
We discuss the application of a variant of  the method of simplest equation  
for obtaining exact traveling wave solutions of a  class of nonlinear partial 
differential equations containing polynomial nonlinearities. As simplest equation 
we use differential equation for a special function that contains as particular 
cases trigonometric and hyperbolic functions as well as the elliptic function of 
Weierstrass and Jacobi. We show that for this case the studied class of nonlinear 
partial differential equations can be reduced to a system of two equations containing 
polynomials of the unknown functions. This system may be further reduced to
a system of nonlinear algebraic equations for the parameters of the solved equation and
parameters of the solution. Any nontrivial solution of the last system leads to a 
traveling wave solution of the solved nonlinear 
partial differential equation. The methodology is illustrated by obtaining solitary
wave solutions for the generalized Korteweg-deVries equation and by obtaining
solutions of the higher order Korteweg-deVries equation.
\end{abstract}
\section{Introduction}
Nonlinearity is an essential feature of many systems in Nature and society
\cite{medio} - \cite{viti11} 

Traveling wave solutions of  nonlinear partial differential equations are
studied much in the last decades \cite{murr}-\cite{kudr90} as they occur in many natural systems 
\cite{scott02}-\cite{ablowitz1}  and because of  existence of various  methods for obtaining such solutions 
\cite{gardner} - \cite{wazwaz1}. Below  we shall consider  the method of simplest equation for 
obtaining exact analytical solutions of nonlinear partial differential equations 
\cite{kudr05x}-\cite{k10} and especially  its version called modified method of simplest equation \cite{vdk10} - \cite{vit11b}. 
Method of simplest equation is based on  a procedure analogous to the first 
step of the test for the Painleve property \cite{k10}, \cite{k15}. In the version of the method 
called  modified method of the simplest equation \cite{vit09x},\cite{vd10}  this procedure is 
substituted by  the concept for the balance equation. Modified method of simplest equation has its roots back in the history (for an
example see \cite{koh}-\cite{vitx2}). Method of simplest equation has
been successfully applied for obtaining exact traveling wave solutions of  numerous
nonlinear PDEs such as versions of generalized Kuramoto - Sivashinsky equation, reaction - diffusion  equation, reaction - telegraph equation, 
generalized Swift - Hohenberg equation and generalized Rayleigh equation, 
generalized Fisher equation, generalized Huxley equation , generalized Degasperis - Procesi equation and b-equation, 
extended Korteweg-de Vries equation , etc.  \cite{kudr05} - \cite{vit14}.
\par
A short summary of the method of simplest equation is as follows. First of all by means 
of an appropriate ansatz (for an example the traveling-wave ansatz) the solved  of nonlinear partial differential equation is reduced to a nonlinear
ordinary differential equation
\begin{equation}\label{i1}
P \left( u, u_{\xi},u_{\xi \xi},\dots \right) = 0
\end{equation}
Then the finite-series solution 
\begin{equation}\label{i2}
u(\xi) = \sum_{\mu=-\nu}^{\nu_1} p_{\mu} [g (\xi)]^{\mu}
\end{equation}
is substituted in (\ref{i1}). $p_\mu$ are coefficients and $g(\xi)$ is solution of
simpler ordinary differential equation called simplest equation. Let the result
of this substitution be a polynomial of $g(\xi)$. Eq. (\ref{i2}) is a solution of Eq.(\ref{i1}) 
if all coefficients of the obtained polynomial of $g(\xi)$ are equal to $0$. This condition 
leads to a system of nonlinear algebraic equations. Each solution of the last  system  leads to a solution of the studied  
nonlinear partial differential equation.
\par 
In this article we consider a large class of (1+1)-dimensional nonlinear partial differential 
equations that are constricted by polynomials of the unknown function and its derivatives. As simplest equation we shall use equation of the kind
$$
\left(\frac{dg}{d\xi}\right)^2 = \sum \limits_{i=0}^m a_i g^i.
$$
The text below is organized as follows. In Sect. 2 we introduce the class
of studied nonlinear partial differential equations and  the used 
class of simplest equations and their solutions. Then we show  that 
any of the nonlinear partial differential equations of the discussed class can be
reduced to a system of two equations containing polynomials of the unknown function.
These polynomials can be obtained on the basis of addition and multiplication of 
some basic polynomials connected to the derivatives of the solved 
nonlinear partial differential equation. In section 3 we calculate some of the most used basic polynomials.
In Sect. 4 the methodology is illustrated by application for obtaining solitary wave solutions of 
\begin{itemize}
\item generalized Korteweg-deVries equation;
\item higher order Korteweg-deVries equation. 
\end{itemize}
Several concluding remarks are summarized in Sect. 5.
\section{Formulation of the method}
\subsection{Proof of the basic theorem}
Let us consider a nonlinear PDE with nonlinearities that are polynomials of
the unknown function $h(x,t)$ and its derivatives. We search 
solution of the kind
\begin{equation}\label{a1}
h(x,t) = h(\xi); \ \    \xi = \mu x + \nu t
\end{equation}
where $\mu$ and $\nu$ are parameters. The basis of our search will be a
solution $g(\xi)$ of a certain simplest equation. Hence 
\begin{equation}\label{a2}
h=f[g(\xi)]
\end{equation}
$h$ from Eq.({\ref{a2}) is a composite function. For the $n$-th derivative of
$h$ we have the Faa di Bruno formula \cite{fdb}
\begin{equation}\label{a3}
h_{(n)} = \sum \limits_{k=1}^n f_{(k)} \sum \limits_{p(k,n)} n! \prod_{i=1}^n
\frac{g_{(i)}^{\lambda_i}}{(\lambda_i!)(i!)^{\lambda_i}}
\end{equation}
where
\begin{itemize}
\item $h_{(n)} = \frac{d^n h}{dx^n}$;
\item $f_{(k)} = \frac{d^k f}{ dg^k}$;
\item $g_{(i)} = \frac{d^i g}{dx^i}$;
\item $p(n,k) = \{ \lambda_1, \lambda_2, \dots, \lambda_n \}$: set of
numbers such that 
\begin{equation}\label{a4}
\sum \limits_{i=1}^n \lambda_i = k; \sum_{i=1}^{n} i \lambda_i = n .
\end{equation}
\end{itemize}
Further we shall concentrate on $f_{(k)}$ and $g_{(i)}$.
\par 
Let us now assume that $f$ is a polynomial of $g$. Then
\begin{equation}\label{a5}
f = \sum \limits_{r=0}^q b_r g^r
\end{equation}
Let us consider the derivative $f_{(k)}$. 
If $k>r$ this derivative is $0$. The derivative
is non-zero if $k \le r$. We shall use the function $\Theta_{rk}$ with
the following definition
\begin{equation}\label{a6}
\Theta_{rk} = \begin{cases} 
      0 & r < k \\
      1 & r \ge k  
   \end{cases}
\end{equation}
Then the derivative $f_{(k)}$ can be written as
\begin{equation}\label{a7}
f_{(k)} = \sum \limits_{r=0}^q \Theta_{rk} \frac{r!}{(r-k)!} b_r g^{r-k}
\end{equation}
The derivative $g_{(i)}$ is connected to the simplest equation.
In general we can use the following simplest equation with polynomial nonlinearity
\begin{equation}\label{a8}
g_{(k)}^l = \left( \frac{d^k g}{d \xi^k} \right)^l = \sum_{j=0}^m a_j g^j
\end{equation}
The solution of this equation defines certain special function
$V_{a_0,a_1,\dots,a_m}(\xi; k,l,m$) where
\begin{itemize}
\item $k$: order of derivative of $g$;
\item $l$: degree of derivative in the defining ODE;
\item $m$: highest degree of the polynomial of $g$ in the defining ODE.
\end{itemize}
This special function has very interesting properties as its particular cases
are the trigonometric, hyperbolic, elliptic functions of Jacobi, etc.
Below we shall use the function $V_{a_0,a_1,\dots, a_m}(\xi;1,2,m)$ which is solution of the
simplest equation
\begin{equation}\label{a9}
g_{(1)}^2 = \left( \frac{dg}{d \xi} \right)^2 = \sum_{j=0}^m a_j g^j
\end{equation}
We let $m$ undetermined for now.
\par
If $g_{(1)}^2$ is given by Eq.(\ref{a9}) then what is the relationship for the
derivative $g_{(i)}$?
\begin{lemma}
If $g_{(1)}^2$ is given by Eq. (\ref{a9}) then the following relationship
holds for the derivative $g_{(i)}$:
$$
g_{(i)}^{\lambda_i} = [A_i(g)]^{\lambda_i} O_{i,\lambda_i}(g,g_{(1)}) \Omega_{\lambda_i}(g_{(1)})
$$
where 
\begin{itemize}
\item $\lambda_i$ is a non-negative integer;
\item $A_i(g)$: polynomial of $g$;
\item 
\begin{equation*} 
O_{i,\lambda_i}(g,g_{(1)}) = \begin{cases} 
      1 & i - \textrm{even}, \ \lambda_i - \textrm{even}  \\
      \frac{1}{g_{(1)}} & i - \textrm{even}, \ \lambda_i - \textrm{odd}  \\
      (\sum \limits_{j=0}^m a_j g^j)^{[\lambda_i/2]} & i - \textrm{odd}, \ \lambda_i - \textrm{even}\\ 
      (\sum \limits_{j=0}^m a_j g^j)^{[\lambda_i/2]} & i - \textrm{odd}, \ \lambda_i - \textrm{odd}\\ 
   \end{cases}
\end{equation*}
where $[\lambda_i/2]$ denotes the integer part of $\lambda_i/2$;
\item
\begin{equation*}
\Omega_{i}(g_{(1)}) = \begin{cases} 
      1 = g_{(1)}^0 & i - \textrm{even}  \\
      g_{(1)} & i - \textrm{odd}   
   \end{cases}
\end{equation*}
\end{itemize}
\end{lemma}
\begin{proof}
It is easy to show (by direct differentiation and by
induction) that if $g_{(1)}^2$ is given by Eq. (\ref{a9}) then
\begin{equation}\label{a10}
g_{(2n)}=A_{2n}(g); \ \ \ \ g_{(2n+1)} = A_{2n+1}(g) g_{(1)}
\end{equation} 
where $A_{2n}(g)$ and $A_{2n+1}(g)$ are polynomials of $g$. 
\par
Let then
\begin{equation}\label{a11}
\Omega_{i}(g_{(1)}) = \begin{cases} 
      1 = g_{(1)}^0 & i - \textrm{even}  \\
      g_{(1)} & i - \textrm{odd}   
   \end{cases}
\end{equation}
Hence Eq.(\ref{a10}) can be written as
\begin{equation}\label{a12}
g_{(i)} = A_{i}(g) \Omega_{i}(g_{(1)})
\end{equation}
Thus
\begin{equation}\label{a13}
g_{(i)}^{\lambda_i} = [A_{i}(g)]^{\lambda_i} [\Omega_{i}(g_{(1)})]^{\lambda_i}
\end{equation}
Now from the definition of $\Omega_{i}$ - Eq.(\ref{a11}) we see that
$\Omega_{i}^{\lambda_{i}}$ is equal to $1$ if $i$ is even.
$\Omega_{i}^{\lambda_{i}}$
is equal to $g_{(1)}^{\lambda_i}$ if $i$ is odd. But from Eq.(\ref{a10}) for the case
of odd $\lambda_i$
\begin{equation}\label{a14}
g_{(1)}^{\lambda_i} =g_{(1)}^{2([\lambda_i/2]) + 1} = g_{(1)} (g_{(1)}^2)^{[\lambda_i/2]} 
\end{equation} 
In addition from Eq.(\ref{a10}) it follows that  $g_{(1)}^2$ is a polynomial of $g$ (this also can be seen
from the definition of the simplest equation (\ref{a9})). Thus we can define
\begin{equation}\label{a15} 
O_{i,\lambda_i}(g,g_{(1)}) = \begin{cases} 
      1 & i - \textrm{even}, \ \lambda_i - \textrm{even}  \\
      \frac{1}{g_{(1)}} & i - \textrm{even}, \ \lambda_i - \textrm{odd}  \\
      (\sum \limits_{j=0}^m a_j g^j)^{[\lambda_i/2]} & i - \textrm{odd}, \ \lambda_i - \textrm{even}\\ 
      (\sum \limits_{j=0}^m a_j g^j)^{[\lambda_i/2]} & i - \textrm{odd}, \ \lambda_i - \textrm{odd}\\ 
   \end{cases}
\end{equation}
and then
\begin{equation}\label{a16}
[\Omega_{i}(g_{(1)})]^{\lambda_{i}} = O_{i,\lambda_i}(g,g_{(1)}) \Omega_{\lambda_i}(g_{(1)})
\end{equation}
Substitution of Eq.(\ref{a16}) in Eq.(\ref{a13}) leads us to the relationship we want to prove
\begin{equation}\label{a17}
g_{(i)}^{\lambda_i} = [A_i(g)]^{\lambda_i} O_{i,\lambda_i}(g,g_{(1)})
\Omega_{\lambda_i}(g_{(1)})
\end{equation}
\end{proof}
\par 
Now we are in position to prove
\begin{theorem}
If $g_{(1)}^2$ is given by Eq.(\ref{a9}) and $f$ is a polynomial of $g$ given by Eq.(\ref{a5})
then for $h[f(g)]$ the following relationship holds
$$
h_{(n)} = K_n(q,m)(g) + g_{(1)} Z_n(q,m) (g)
$$
where $K_n(q,m)(g)$ and $Z_n(q,m)(g)$ are polynomials of the function $g(\xi)$.
\end{theorem}
\begin{proof}
The substitution of Eqs. (\ref{a17}) and (\ref{a7}) in Eq.(\ref{a3}) leads to
the following relationship for $h_{(n)}$
\begin{equation}\label{a18}
h_{(n)} = \sum \limits_{k=1}^n \sum \limits_{r=0}^q \Theta_{rk} \frac{r!}{(r-k)!}
b_r g^{r-k} \sum \limits_{p(n,k)} n! \prod \limits_{i=1}^n
\frac{[A_i(g)]^{\lambda_i} O_{i,\lambda_i}(g,g_{(1)}) \Omega_{\lambda_i}(g_{(1)})}{(
\lambda_i!)(i!)^{\lambda_i}}
\end{equation}
Now let us show that the relationship (\ref{a18}) can be written as
\begin{equation}\label{a19}
h_{(n)} = K_n(q,m)(g) + g_{(1)} Z_n(q,m) (g)
\end{equation}
where $K_n(q,m)(g)$ and $Z_n(q,m)(g)$ are polynomials of the function $g(\xi)$.
In oder to show this we rewrite Eq.(\ref{a18}) as 
\begin{equation}\label{a20}
h_{(n)} = \sum \limits_{k=1}^n \sum \limits_{r=0}^q \Theta_{rk} \frac{r!}{(r-k)!}
b_r g^{r-k} \sum \limits_{p(n,k)} n! \left[ \prod \limits_{i=1}^n
\frac{[A_i(g)]^{\lambda_i}}{(
\lambda_i!)(i!)^{\lambda_i}} \right] \prod \limits_{i=1}^n O_{i,\lambda_i}(g,g_{(1)}) \Omega_{\lambda_i}(g_{(1)})
\end{equation}
and consider $P=\prod \limits_{i=1}^n O_{i,\lambda_i}(g,g_{(1)}) \Omega_{\lambda_i}(g_{(1)})$. This
product is equal of polynomial of $g$ multiplied by $g_{(1)}^\sigma$ where $\sigma$ is an 
integer. There are two possibilities for $\sigma$:
\begin{itemize}
\item $\sigma$: \emph{even} . Then 
$(g_{(1)}^2)^{\sigma/2}$ is a polynomial of $g$ according to
Eq.(\ref{a9})
\item $\sigma$: \emph{odd}. Then $g_{(1)}^\sigma$ is equal to 
$g_{(1)}(g_{(1)}^2)^{[\sigma]/2}$ and according to Eq.(\ref{a9}) this is equal to
 $g_{(1)}$ multiplied by a polynomial of $g$.
\end{itemize} 
Because of all above there will be two kinds of
terms in $h_{(n)}$: terms that are polynomials of $g$ and terms that
contain $g_{(1)}$ multiplied by a polynomial of $g$.  Collecting the
two kinds of terms we arrive at Eq.(\ref{a19}). 
\end{proof}
Let us  note that for some values of $n$ one of the  polynomials $K_n(q,m)$ or $Z_n(q,m)$ can 
be equal to $0$. 
\subsection{Formulation of the method}
On the basis of all above the modified method of simplest equation based on
simplest equation (\ref{a9}) can be formulated as follows:
\begin{enumerate}
\item We consider a nonlinear PDE $E^*(u(x,t), u_x(x,t), u_t(x,t), u_{xt}(x,t),
\dots)=0$ where $E^*$ is a polynomial of $u(x,t)$ and its derivatives. We search for
solutions of this equation based on the ansatz $\xi = \mu x + \nu t$  where $\alpha$ and $\beta$
are parameters.
\item The ansatz $\xi = \mu x + \nu t$ reduces the nonlinear PDE to the ODE $E(h,h_{(1)},\dots) = 0$
where $E$ is a polynomial of $h$ and its derivatives.
\item
We assume $h$ and $g$ are given by the Eqs. (\ref{a5}) and (\ref{a9}).
Substitution of the relationships in the equation $E=0$ reduces any 
derivative of this equation to a term of the kind (\ref{a19}).
\item
As the terms in $E=0$ are polynomials of $h$ and its derivatives then
the equation reduces to a polynomial containing $g$ , $g_{(1)}$, $g_{(1)}^2$,
$\dots$: $E = W^*_0(g)+W^*_1 (g) g_{(1)} + W^*_2(g) g_{(1)}^2+ \dots$.
Eq.(\ref{a9}) reduces the higher degrees ($n>1$) of $g_{(1)}$ to a polynomial of
$g$ (for even $n$) or to a polynomial of $g$ multiplied by $g_{(1)}$ (for odd
$n$). Thus equation $E=0$ is reduces to
\begin{equation}\label{a21}
E = W_0(g) + W_1(g) g_{(1)} =0
\end{equation} 
where $W_{0,1}(g)$ are polynomials of $g$.
\item
In order to obtain a nontrivial solution of Eq.(\ref{a21}) we have to balance
the highest degree of the polynomial $W_0$ (i.e., to ensure that there are at least
two terms that contain the highest degree of $W_0$). The same has to be made for
the polynomial $W_1$. As a result we obtain one or two relationships among
the parameters of the equation and parameters of the solution. These equations
are called balance equations. Balance equations fix parameters $q$ and $m$.
\item
Further we  set to $0$ all coefficients
of the polynomials $W_0(g)$ and $W_1(g)$. The result is a system of nonlinear
algebraic equations that contains the parameters of the equation, parameters of
the solution (\ref{a5}) and parameters of the simplest equation (\ref{a9}).
\item
Any nontrivial solution of the above system of algebraic equations (if it
exists) leads to a traveling wave solution of the nonlinear PDE $E^*=0$.
\end{enumerate}
\section{Calculation of some of polynomials $K_n$ and $Z_n$ from Eq.(\ref{a19})}
The derivatives $h_{(1)}$, $h_{(2)}$, $h_{(3)}$, $h_{(4)}$, $h_{(5)}$, $h_{(6)}$
and $h_{(7)}$ are much used in the model nonlinear partial differential equations.
Below we shall calculate the polynomials $K_n$ and $Z_n$ connected to these derivatives.
\par
Let us mention first that the polynomials $K_{n+1}$ and $Z_{n+1}$ are connected to the
polynomials $K_n$ and $Z_n$. This relationship can be obtained on the
basis of the relationship
\begin{equation}\label{b1}
h_{(n+1)} = \frac{d}{d \xi} h_{(n)} 
\end{equation}
Substitution of Eq.(\ref{a19}) in Eq.(\ref{b1}) leads to
\begin{equation}\label{b2}
h_{(n+1)}=\left[Z_n g_{(2)} + \frac{dZ_n}{dg}g_{(1)}^2 \right] + \frac{dK_n}{dg} g_{(1)} 
\end{equation}
Taking in account Eq.(\ref{a9}) we obtain
\begin{eqnarray}\label{b3}
 K_{n+1} &=& \frac{Z_n}{2} \sum \limits_{j=0}^m j a_j g^{j-1} + \frac{dZ_n}{dg}
 \sum \limits_{j=0}^m a_j g^j \nonumber \\
  Z_{n+1} &=& \frac{dK_n}{dg}
\end{eqnarray}
Hence we need to calculate only $K_1$ and $Z_1$ and  then we can obtain $K_n$ and
$Z_n$, $n=2,3,\dots$ by the recurrence equations (\ref{b3}).
\par
We can write
\begin{eqnarray}\label{b3_0}
K_0 &=& \sum \limits_{r=0}^q b_r g^r \nonumber \\
Z_0 &=& 0
\end{eqnarray}
From Eq.(\ref{a18}) we obtain
\begin{equation}\label{b4}
K_1 = 0; \ \ Z_1 = \sum \limits_{r=0}^q r b_r g^{r-1}
\end{equation}
Then
\begin{eqnarray}\label{b5}
K_2 &=& \sum \limits_{r=0}^q \sum \limits_{j=0}^m \left[ \frac{1}{2} jr + r(r-1) 
\right]a_jb_r g^{j+r-2}\nonumber \\
Z_2 &=& 0.
\end{eqnarray}
\begin{eqnarray}\label{b6}
K_3 &=& 0; \nonumber \\
Z_3 &=&  \sum \limits_{r=0}^q \sum \limits_{j=0}^m \left[ \frac{1}{2} jr + r(r-1)
 \right] (j+r-2)a_jb_r g^{j+r-3}\nonumber \\ 
\end{eqnarray}
\begin{eqnarray}\label{b7}
K_4&=& \sum \limits_{r=0}^q \sum \limits_{j=0}^m \sum \limits_{u=0}^m 
\left[\left(\frac{1}{2} jr+r(r-1) \right) (j+r-2) \left( \frac{1}{2} u + j+r
-3 \right) \right] a_jb_r a_u g^{j+r+u-4}
\nonumber \\
Z_4&=&0.
\end{eqnarray}
\begin{eqnarray}\label{b8}
K_5 &=& 0; \nonumber \\
Z_5 &=& \sum \limits_{r=0}^q \sum \limits_{j=0}^m \sum \limits_{u=0}^m 
\left[\left(\frac{1}{2} jr+r(r-1) \right) (j+r-2) \left( \frac{1}{2} u + j+r
-3 \right) \right] (j+ \nonumber \\
&& r+u-4) a_jb_r a_u g^{j+r+u-5}
\end{eqnarray}
\begin{eqnarray}\label{b9}
K_6 &=& \sum \limits_{r=0}^q \sum \limits_{j=0}^m \sum \limits_{u=0}^m 
\sum \limits_{v=0}^m
\bigg[ \left(\frac{1}{2} jr+r(r-1) \right) (j+r-2) \left( \frac{1}{2} u + j+r
-3 \right) (j+ \nonumber \\
&& r+u-4) \bigg] \bigg(\frac{1}{2}v+ j+r+u-5\bigg) a_j b_r a_u a_v g^{j+r+u+v-6}\nonumber \\
Z_6 &=&0.
\end{eqnarray}
\begin{eqnarray}\label{b10}
K_7 &=& 0 ; \nonumber \\
Z_7 &=& \sum \limits_{r=0}^q \sum \limits_{j=0}^m \sum \limits_{u=0}^m 
\sum \limits_{v=0}^m
\bigg[ \left(\frac{1}{2} jr+r(r-1) \right) (j+r-2) \left( \frac{1}{2} u + j+r
-3 \right) (j+ \nonumber \\
&& r+u-4) \bigg] \bigg(\frac{1}{2}v+ j+r+u-5\bigg)(j+r+u+v-6) a_j b_r a_u a_v 
g^{j+r+u+v-7} \nonumber \\
\end{eqnarray}
etc.
\par
For the practical application of the modified method of simplest equation
we need to calculate the maximum grade of polynomials in $h_{(n)}$.
As we have seen  above in the text $h_{(n)}$ consists of two kinds of terms: polynomial
of $g$ plus another polynomial of $g$ multiplied by $g_{(1)}$.  
We note that the above maximum grades have to be non-negative. Thus the
obtained relationships below hold when the corresponding maximum degree is 
$max \ge 0$.
\par
By the method of mathematical induction we can prove that:
\begin{itemize}
\item \emph{For $h_{(2\sigma)}$:} Maximum grade of polynomial  $K_{2\sigma}(g)$ is
$max=q+\sigma(m-2)$ where $\sigma=1,2,\dots$. Maximum grade of the polynomial
$Z_{2\sigma}(g)$  is 0.
\item \emph{For $h_{(2\sigma +1)}$:} Maximum grade of polynomial of
$K_{2\sigma+1}(g)$ is
$0$. Maximum grade of the polynomial
of $Z_{2\sigma+1}(g)$ is $max=q+\sigma(m-2)-1$ where $\sigma =
0,1,\dots$.
\end{itemize}
\section{Examples}
\subsection{Generalized Korteweg-deVries equation}
We shall consider the equation
\begin{equation}\label{c1}
\frac{\partial u}{\partial t} + Au^p \frac{\partial u}{\partial x} +
 \frac{\partial^3 u}{\partial x^3} =0
\end{equation}
where $A$ is a parameter and $p$ is a positive integer number. This equation
has different applications as for an example in the electrohydrodynamics
\cite{perelman}.
For the case when $p$ is a positive integer Eq.(\ref{b1})
is called generalized Korteweg-deVries equation. It is obtained by 
addition of the dispersion term 
$\dfrac{\partial^3 u}{\partial x^3}$ to the nonlinear convective wave equation
$\dfrac{\partial u}{\partial t} + A u^p \dfrac{\partial u}{\partial x} =0$.
\par
We search for  solutions of Eq.(\ref{c1})  of the kind
$u=h[f(g(\xi))]$ where $\xi = \alpha x + \beta t$, $g$ is solution of the
simplest equation  (\ref{a9}) and $f$ is given by Eq.(\ref{a5}). 
The substitution of $u=h[f(g(\xi))]$ in Eq.(\ref{c1}) leads to  equation
of the kind (\ref{a21}) where
\begin{eqnarray}\label{c2}
W_0(g) &=& \nu K_1(g) + \mu A K_0(g)^p K_1(g) +  \mu^3 K_3(g) \nonumber \\
W_1(g) &=& \nu Z_1(g) + \mu A K_0(g)^p Z_1(g) +  \mu^3 Z_3(g)
\end{eqnarray}
As $K_1 = K_3 = 0$ there is no need to balance the relationship for $W_0(g)$.
The relationship for $W_1(g)$ has to be balanced however. The resulting balance
equation is
\begin{equation}\label{c3}
m = 2 + pq
\end{equation}
Let us consider the case $q=1$, $m=2+p$. Then from Eqs.(\ref{a5}) and (\ref{a9})
one obtains
\begin{equation}\label{c4}
h=b_0 + b_1 g; \ \ \ g_{(1)}^2 = \sum \limits_{j=0}^{2+p} a_j g^j.
\end{equation} 
Substituting corresponding relationships for the polynomials $Z_1$, $Z_3$
and $K_0$ in the second of equations (\ref{c2}) we obtain the following
system of nonlinear algebraic relationships among the parameters of Eq.(\ref{c1})
and the parameters of the solution:
\begin{equation}\label{c4}
\nu b_1 \delta_{0,k} + {p \choose k} \mu  A b_0^{p-k} b_1^{k+1} + \frac{1}{2}
\mu^3  (k+1)(k+2) a_{k+2}b_1 =0, \ \ k=0,\dots,p.
\end{equation}
where $\delta$ is the delta-symbol of Kronecker.
\par
Solution of Eqs.(\ref{c4}) can be obtained when $b_0=0$. Then
\begin{equation}\label{c5}
\nu b_1 \delta_{0,k} + { p \choose k} \mu A b_1^{k+1} \delta_{k,p} + \frac{1}{2}
\mu^3  (k+1)(k+2) a_{k+2}b_1 =0, \ \ k=0,\dots,p.
\end{equation}
The system (\ref{c5}) is reduced to
\begin{eqnarray}\label{c6}
k&=&0: \ \ \ \nu + \mu^3  a_2 =0 \nonumber \\
k&=&1:  \ \ \ a_3 = 0 \nonumber \\
&&\dots  \dots \nonumber \\
k&=&p-1: \ \ \ a_{p+1}=0 \nonumber \\
k&=&p:  \ \ \ A b_1^p + \frac{1}{2} \mu^2  (p+1)(p+2) a_{p+2} = 0 .
\end{eqnarray}
The system (\ref{c6}) has a solution: 
\begin{equation}\label{c7}
a_2 = - \frac{\nu}{\mu^3 }; \ \ \ a_{p+2} = - \frac{2Ab_1^p}{\mu^2
(p+1)(p+2)}
\end{equation}
Hence the simplest equation Eq.(\ref{a9}) becomes
\begin{equation}\label{c8}
g_{(1)}^2 = -\frac{\nu}{\mu^3 A} g^2 - \frac{2b_1^p}{\mu^2
A(p+1)(p+2)} g^{p+2}
\end{equation}
A solution of this simplest equation is as follows. The function
\begin{equation}\label{c9}
g(\xi) = \frac{\Omega}{\cosh^\omega(\xi)}
\end{equation}
(where $\omega$ and $\Omega$ are parameters) is solution of the equation
\begin{equation}\label{c10}
g_{(1)}^2 = \omega^2 g^2 - \frac{\omega^2}{\Omega^{2/\omega}} g ^{2+ 2/\omega}
\end{equation}
Hence if
\begin{equation}\label{c11}
\omega = \frac{2}{p}; \ \ \mu^2 =   \frac{p^2 A \Omega^p b_1^p}{2 
(p+1)(p+2)} ;
 \ \ \nu =  - \frac{4 \mu^3}{p^2}
\end{equation}
then
\begin{equation}\label{c12}
u(\xi) = \frac{\Omega b_1}{\cosh^{2/p}(\xi)}; \xi = \mu x + \nu t
\end{equation}
is solution of  Eq.(\ref{c1}).
\par
Let us note that according to Eqs.(\ref{c9}) and (\ref{c10})
\begin{equation}\label{c12x}
V_{0,0,4/p^2,0,\dots,a_{p+2}=-\frac{4}{p^2 \Omega^p}}(\xi;1,2,2+p) =
\frac{\Omega}{\cosh^{2/p}(\xi)}
\end{equation}
\par
Let us consider the particular case $p=1$, $A=-6$. Then Eq.(\ref{c1}) is
reduced to the classical Korteweg-deVries equation. Setting $b_1 = 1$ and $\Omega = -2$
we obtain $\mu=1$, $\nu = -4$. Thus the solution (\ref{c12}) reduces to
the one-soliton solution of the equation of Korteweg-deVries
\begin{equation}\label{c13}
u(x,t) = -\frac{2}{\cosh^2(x-4t)}
\end{equation}
Let us stress the following.
Above $p$ was arbitrary non-negative integer. Now we shall show that $p$ can be
arbitrary non-zero real number. Namely we shall prove
\begin{prop}
Eq. (\ref{c12}) is solution of Eq.(\ref{c1}) for arbitrary real nonzero value of
$p$.
\end{prop}
\begin{proof}
Let us substitute Eq.(\ref{c12}) in Eq.(\ref{c1}) where parameters $\alpha$ and
$\beta$ are given by Eq.(\ref{c11}) and $p$ is an arbitrary nonzero real number.
Eq.(\ref{c1}) is satisfied. Hence Eq.(\ref{c12}) is solution of Eq.(\ref{c1})
for arbitrary real $p \ne 0$.
\end{proof}
Let for an example $p=3/2$. Then a solitary wave solution of the equation
\begin{equation}\label{c14}
\frac{\partial u}{\partial t} + Au^{3/2} \frac{\partial u}{\partial x} +
 \frac{\partial^3 u}{\partial x^3} =0
\end{equation}
is
\begin{equation}\label{c15}
u(x,t) = \Omega b_1 \cosh^{-4/3} \left[\left( \frac{9}{70} A \Omega^{3/2} b_1^{3/2} \right)^{1/2} x -
\frac{16}{9} \left( \frac{9}{70} A \Omega^{3/2} b_1^{3/2} \right)^{3/2} t \right]
\end{equation}
\subsection{Application of methodology to the higher order Korteweg-deVries
equation}
Let us apply the above methodology to a more complicated equation such as the
second-order Korteweg-deVries equation
\begin{equation}\label{d1}
\frac{\partial u}{\partial t} + \frac{\partial u}{\partial x} + \alpha_0 u 
\frac{\partial u}{\partial x} + \alpha_1 \frac{\partial u}{\partial x}
\frac{\partial^2 u}{\partial x^2}+ \alpha_2 u \frac{\partial^3 u}{\partial x^3}
+ \alpha_3 u^2 \frac{\partial u}{\partial x} + \alpha_4 \frac{\partial^3
u}{\partial x^3} + \alpha_5 \frac{\partial^5 u}{\partial x^5} =0
\end{equation}
\par 
This equation is known also as Olver equation \cite{o1} and it is a second
order equation for description of shallow water waves (the first order
equation is the famous Korteweg-deVries equation). 
Eq.(\ref{d1}) has been used also as a model equation for nonlinear waves
in a liquid with gas bubbles \cite{ks14a,ks14b}
From the point of view of the 
method of simplest equation (based on the first step of a test for Painleve property)
the equation was discussed in \cite{ks14a}-\cite{k12} and various solutions expressed by the elliptic
function of Weierstrass have been obtained there. Below we shall discuss Eq.(\ref{d1})
from the point of view of the modified method of simplest equation (based on the concept
for balance equation). We stress again that  the modified method of simplest equation is a version of the 
method of simplest equation and then some of the obtained below solutions (especially for the 
case $m=3$) will be the same as these obtained in \cite{k12}. After studying the methodologies
from \cite{k12}  and from this paper the reader will have an extensive understanding about
the possibilities for obtaining exact analytical solutions of nonlinear partial 
differential equations on the basis of the method of simplest equation.
\par
We search for a solution of the kind $u(x,t) = u(\xi) = u(\alpha x + \beta t)$ 
(We note that in \cite{k12} $\alpha=1$).
According to the theory above Eq.(\ref{d1}) can be reduced to the following
system of equations 
\begin{eqnarray}\label{d2}
W_0(g)&=&0; \nonumber \\
W_1(g)&=& (\mu + \nu) Z_1(g) + \alpha_0 \mu K_0(g) Z_1(g) + \alpha_1
\mu^3 K_2(g) Z_1(g) + \nonumber \\
&& \alpha_2 \mu^3 K_0(g) Z_3(g) + \alpha_3 \mu
K_0^2(g) Z_1(g) + \alpha_4 \mu^3 Z_3(g) + \nonumber \\
&& \alpha_5 \mu^5 Z_5(g) =0
\end{eqnarray}
The second of Eqs.({\ref{d2}) has to be balanced and the balance is as follows
\begin{enumerate}
\item $m=1$: there is no balance;
\item $m=2$: there is no balance;
\item $m \ge 3$: the equation is balanced if $q=m-2$.
\end{enumerate}
Let us now consider several cases.
\subsection{Case $m=3$}
In this case $q=1$. Then
\begin{equation}\label{d10}
u(\xi) = b_0 + b_1 g(\xi); \ \ \xi = \mu x + \nu t
\end{equation}
and 
\begin{equation}\label{d11}
g_{(1)}^2 = a_0 + a_1 g + a_2 g^2 + a_3 g^3
\end{equation}
In addition we have to solve the system of nonlinear algebraic equations for the
parameters of the solution that can be obtained from the equation $W_1(g) =0$
from Eqs.(\ref{d2}).
\par
We note that the general solution of Eq.(\ref{d11}) is given by the special
function $V_{a_0,a_1,a_2,a_3}(\xi; 1,2,3)$ and for the special case when
$a_2 =0$ and $a_3=4$ we have reduction of $V$ to the elliptic function of
Weierstrass
\begin{equation}\label{d12}
V_{a_0,a_1,0,4}(\xi;1,2,3) = \wp(\xi;a_0,a_1)
\end{equation}
\par
Let us first consider the case of general non-reduced solution $V_{a_0,a_1,a_2,a_3}(\xi; 1,2,3)$ of Eq.(\ref{d11}). 
The system of nonlinear algebraic equations arising from second of the Eqs.(\ref{d2}) is (\ref{d13x}) from the 
appendix A. The solution of the system (\ref{d13x}) is:
\begin{eqnarray}\label{d14}
a_3 &=&   \frac{b_1}{30 \mu^2 \alpha_5} \bigg(-\alpha_1 - 2 \alpha_2 + \sqrt{
(\alpha_1 + 2 \alpha_2)^2 - 40 \alpha_3 \alpha_5} \bigg)\nonumber \\
a_2 &=& - \frac{1}{5 \mu^2 \alpha_5 \bigg(\alpha_1 + \sqrt {(\alpha_1 + 2 \alpha_2)^2 - 40 \alpha_3 \alpha_5}\bigg)} \bigg[ 
\nonumber \\
&& (\alpha_2 b_0 + \alpha_4) \sqrt{(\alpha_1 + 2 \alpha_2)^2 - 40 \alpha_3 \alpha_5} -\alpha_1 \alpha_2 b_0 - 2 \alpha_2^2 b_0 +
\nonumber \\
&& 20 \alpha_3 \alpha_5 b_0 + 10 \alpha_0 \alpha_5 - \alpha_4 \alpha_1 - 2 \alpha_4  \alpha_2 \bigg] \nonumber \\
a_1 &=& \frac{T_1}{T_2} \nonumber \\
T_1 &=& - \frac{2}{5} \bigg[ 25 \mu \alpha_1^2 \alpha_3 \alpha_5 b_0^2 - 6 \mu \alpha_2^4 b_0^2 + 90 \mu \alpha_2^2 \alpha_3 \alpha_5 b_0^2 + 25 \mu \alpha_0 \alpha_1^2 \alpha_5 b_0 + \nonumber \\
&& 30 \mu \alpha_0 \alpha_2^2 \alpha_5 b_0 - 300 \mu \alpha_0 \alpha_3 \alpha_5^2 b_0 + 2 \mu \alpha_1^2 \alpha_2 \alpha_4 b_0 - 2 \mu \alpha_1 \alpha_2^2 \alpha_4 b_0 + \nonumber \\
&& 50 \mu \alpha_0^2 \alpha_5^2 - 35 \mu \alpha_0 \alpha_1 \alpha_4 \alpha_5 - 20 \mu \alpha_0 \alpha_2 \alpha_4 \alpha_5 + \mu \alpha_1^2 \alpha_4^2 - 6 \mu \alpha_2^2 \alpha_4^2 + \nonumber \\
&& 25 \mu \alpha_1^2 \alpha_5 + 50 \mu \alpha_2^2 \alpha_5 - 500 \mu \alpha_3 \alpha_5^2 + 25 \nu \alpha_1^2 \alpha_5 + 50 \nu \alpha_2^2 \alpha_5 - \nonumber \\
&& 500 \nu \alpha_3 \alpha_5^2 + 80 \mu \alpha_3 \alpha_4^2 \alpha_5 + 50 \mu \alpha_1  \alpha_2 \alpha_5 + 50 \nu \alpha_1 \alpha_2 \alpha_5 + \mu \alpha_1^2 \alpha_2^2 b_0^2 - \nonumber \\
&& \mu \alpha_1 \alpha_2^3 b_0^2 - 300 \mu \alpha_3^2 \alpha_5^2 b_0^2 - 12 \mu \alpha_2^3 \alpha_4 b_0 - \mu \alpha_1 \alpha_2 \alpha_4^2  - 20 \mu \alpha_1 \alpha_2 \alpha_3 \alpha_5 b_0^2 + \nonumber \\
&& 15 \mu \alpha_0 \alpha_1 \alpha_2 \alpha_5 b_0 - 70 \mu \alpha_1 \alpha_3 \alpha_4 \alpha_5 b_0 + 120 \mu \alpha_2 \alpha_3 \alpha_4 \alpha_5 b_0 + \nonumber \\
&& \bigg( 6 \mu \alpha_2^2 \alpha_4 b_0  - \mu \alpha_1 \alpha_2^2 b_0^2 - 15 \mu \alpha_0 \alpha_4 \alpha_5 + 3 \mu \alpha_2^3 b_0^2 - \mu a_1 a_4^2 + 3 \mu a_2 a_4^2 + \nonumber \\
&& 25 \mu \alpha_1 \alpha_5 + 25 \nu \alpha_1 \alpha_5 + 25 \mu \alpha_1 \alpha_3 \alpha_5 b_0^2 - 30 \mu \alpha_2 \alpha_3 \alpha_5 b_0^2 + 25 \mu \alpha_0 \alpha_1 \alpha_5 b_0 - \nonumber \\
&& 2 \mu \alpha_1 \alpha_2 \alpha_4 b_0 - 30 \mu \alpha_3 \alpha_4 \alpha_5 b_0 \bigg) \sqrt{(\alpha_1 +  2 \alpha_2)^2 - 40 \alpha_3 \alpha_5} \bigg] \nonumber \\
T_2 &=& \mu^2 \alpha_5 b_1 \bigg[(5 \alpha_1^2 + 3 \alpha-2^2 - 30 \alpha_3 \alpha_5) \sqrt{(\alpha_1 + 2 \alpha_2)^2 - 40 \alpha_3 \alpha_5} + \nonumber \\
&& 10 \alpha_1^2 \alpha_2 + 7 \alpha_1  \alpha_2^2 - 130 \alpha_1 \alpha_3 \alpha_5 - 6 \alpha_2^3 + 60 \alpha_2 \alpha_3 \alpha_5 \bigg] 
\end{eqnarray}
The solution of the higher order Korteweg-deVries equation (\ref{d1}})  is
\begin{eqnarray}\label{d15}
u(\xi) = b_0 + b_1 V_{a_0,a_1,a_2,a_3}(\xi;1,2,3); \ \ \xi = \mu x + \nu t
\end{eqnarray}
where $a_{1,2,3}$ are given by Eqs.(\ref{d14}).
\par
Let us now consider the particular case (\ref{d12}) when the V-function is
reduced to the elliptic function of Weierstrass. In this case we have to set
$a_2=0$; $a_3=4$ in the system (\ref{d13}). The solution of the obtained system of algebraic equations is
\begin{eqnarray}\label{d16}
b_0 &=& - \frac{\alpha_0 \sqrt{(\alpha_1 + \alpha_2)^2 - 40 \alpha_3 \alpha_5}- \alpha_0 \alpha_1 - 2 \alpha_0 \alpha_2 + 4 \alpha_4 \alpha_3}{2 \alpha_3  \bigg( \sqrt{(\alpha_1 + 2 \alpha_2)^2 - 40 \alpha_3 \alpha_5} - \alpha_1 \bigg)} \nonumber \\
b_1 &=&  \frac{3 \mu^2 \bigg( -\alpha_1 - 2 \alpha_2 + \sqrt{(\alpha_1 + 2 \alpha_2)^2 - 40 \alpha_3 \alpha_5}\bigg) }{\alpha_3} \nonumber \\
\end{eqnarray}
\begin{eqnarray*}
a_1 &=& \frac{T_3}{T_4} \nonumber \\
T_3 &=& - \frac{1}{3} \bigg[ -\mu \alpha_0^2 \alpha_1^2 - 2 \mu \alpha_0^2 \alpha_1 \alpha_2 + 20 \mu \alpha_0^2 \alpha_3 \alpha_5 - 8 \mu \alpha_0 \alpha_2 \alpha_3 \alpha_4 + 8 \mu \alpha_3^2 \alpha_4^2 +\nonumber \\
&&  + 4 \mu \alpha_1^2 \alpha_3 + 8 \mu \alpha_1 \alpha_2 \alpha_3 + 8 \mu \alpha_2^2 \alpha_3 - 80 \mu \alpha_3^2 \alpha_5 + 4 \nu \alpha_1^2 \alpha_3 + 8 \nu \alpha_1 \alpha_2 \alpha_3 + \nonumber \\
&& 8 \nu \alpha_2^2 \alpha_3 - 80 \nu \alpha_3^2 \alpha_5 +
\bigg(\mu \alpha_0^2 \alpha_1 - 4 \mu \alpha_1 \alpha_3 - 4 \nu \alpha_1 \alpha_3 \bigg) 
\sqrt{(\alpha_1 + 2 \alpha_2)^2 - 40 \alpha_3 \alpha_5} \bigg]\nonumber \\
T_4 &=& \mu^5 \bigg( \sqrt{(\alpha_1 + 2 \alpha_2)^2 - 40 \alpha_3 \alpha_5} - \alpha_1 \bigg)^2 \bigg( \alpha_1 \sqrt{(\alpha_1 + 2 \alpha_2)^2 - 40 \alpha_3 \alpha_5}-\nonumber \\
&& \alpha_1^2 - 2 \alpha_1 \alpha_2 + 12 \alpha_3 \alpha_5\bigg)
\end{eqnarray*}
The solution of Eq.(\ref{d1}) becomes
\begin{eqnarray}\label{d17}
u(\xi) &=& - \frac{\alpha_0 \sqrt{(\alpha_1 + 2 \alpha_2)^2 - 40 \alpha_3 \alpha_5}- \alpha_0 \alpha_1 - 2 \alpha_0 \alpha_2 + 4 \alpha_4 \alpha_3}{2 \alpha_3  \bigg( \sqrt{(\alpha_1 + 2 \alpha_2)^2 - 40 \alpha_3 \alpha_5} - \alpha_1 \bigg)} + \nonumber \\
&& \frac{3 \mu^2 \bigg( -\alpha_1 - 2 \alpha_2 + \sqrt{(\alpha_1 +  2 \alpha_2)^2 - 40 \alpha_3 \alpha_5}\bigg) }{\alpha_3}
 \wp (\xi; a_0, a_1) \nonumber \\
&& \xi = \mu x + \nu t.
\end{eqnarray}
where $a_1$ is given by the corresponding relationship from Eqs.(\ref{d16}).
\par 
Let us now consider as simplest equation Eq.(\ref{c10}) for the case $\omega=2$,
namely
\begin{equation}\label{dd1}
g_{(1)}^2 = 4 g^2 - \frac{4}{\Omega} g^3
\end{equation}
In this case $a_0 = a_1 =0$, $a_2 =4$ and $a_3 = -\frac{4}{\Omega}$. The solution of Eq.(\ref{dd1}) is
\begin{equation}\label{dd2}
g(\xi) = \frac{\Omega}{\cosh^2(\xi)} = V_{0,0,4,-\frac{4}{\Omega}}(\xi; 1,2,3)
\end{equation}
One solution of the system of nonlinear algebraic equations for this case is
\begin{eqnarray}\label{dd3}
b_0 &=& - \frac{1}{\alpha_3 \sqrt{(\alpha_1 + 2 \alpha_2)^2 - 40 \alpha_3 \alpha_5} + \alpha_1} \bigg[ \big( 4 \mu^2 \alpha_1
+ 4 \mu^2 \alpha_2 + \nonumber \\
&& \alpha_0 \big) \sqrt{(\alpha_1 + 2 \alpha_2)^2 - 40 \alpha_3 \alpha_5}  + 4 \mu^2 \alpha_1^2 + 12 \mu^2 \alpha_1 \alpha_2 + \nonumber \\
&& 8 \mu^2 \alpha_2^2 - 80 \alpha_5 \mu^2 \alpha_3 + \alpha_0 \alpha_1 + 2 \alpha_0 \alpha_2 - 4 \alpha_4 \alpha_3 \bigg]\nonumber \\
\end{eqnarray}
\begin{eqnarray*}
b_1 &=& \frac{3 \mu^2}{\alpha_3 \Omega} \bigg( \alpha_1 + 2 \alpha_2 + \sqrt{(\alpha_1 + 2 \alpha_2)^2 - 40 \alpha_3 \alpha_5} \bigg)  \nonumber \\
\nu &=& -\frac{\mu}{4 \alpha_3 \big( \sqrt{(\alpha_1 + 2 \alpha_2)^2 - 40 \alpha_3 \alpha_5} \alpha_1 + \alpha_1^2 + 2 \alpha_1 \alpha_2 + 2 \alpha_2^2 - 20 \alpha_3 \alpha_5 \big)} \bigg[ \nonumber \\
&& \big( \alpha_0^2 \alpha_1 + 4 \alpha_1 \alpha_3 + 16 \mu^4 \alpha_1^3 + 32 \mu^4 \alpha_1^2 \alpha_2 + 16 \mu^4 \alpha_1 \alpha_2^2 - \nonumber \\
&& 256 \mu^4 \alpha_1 \alpha_3 \alpha_5 \big) \sqrt{(\alpha_1  + 2 \alpha_2)^2 - 40 \alpha_3 \alpha_5} +64 \mu^4 \alpha_1^3 \alpha_2 + 80 \mu^4 \alpha_1^2 \alpha_2^2 - \nonumber \\
&& 576 \mu^4 \alpha_1^2 \alpha_3  \alpha_5 + 32 \mu^4 \alpha_1 \alpha_2^3 - 512 \mu^4 \alpha_1 \alpha_2 \alpha_3 \alpha_5 - 192 \mu^4 \alpha_2^2 \alpha_3 \alpha_5 + 1920 \mu^4 \alpha_3^2 \alpha_5^2 -\nonumber \\
&& \alpha_0^2 \alpha_1^2 - 2 \alpha_0^2 \alpha_1 \alpha_2 + 20 \alpha_0^2 \alpha_3 \alpha_5 - 8 \alpha_0 \alpha_2 \alpha_3 \alpha_4 + 8 \alpha_3^2 \alpha_4^2 + 4 \alpha_1^2 \alpha_3 + 8 \alpha_1 \alpha_2 \alpha_3 + \nonumber \\
&& 8 \alpha_2^2 \alpha_3 - 80 \alpha_3^2 \alpha_5\bigg]
\end{eqnarray*}
and the corresponding solitary-wave solution of Eq.(\ref{d1}) is
\begin{eqnarray}\label{dd4}
u(\xi) &=& - \frac{1}{\alpha_3 \sqrt{(\alpha_1 + 2 \alpha_2)^2 - 40 \alpha_3 \alpha_5} + \alpha_1} \bigg[ \big( 4 \mu^2 \alpha_1
+ 4 \mu^2 \alpha_2 + \nonumber \\
&& \alpha_0 \big) \sqrt{(\alpha_1  + 2 \alpha_2)^2 - 40 \alpha_3 \alpha_5}  + 4 \mu^2 \alpha_1^2 + 12 \mu^2 \alpha_1 \alpha_2 + \nonumber \\
&& 8 \mu^2 \alpha_2^2 - 80 \alpha_5 \mu^2 \alpha_3 + \alpha_0 \alpha_1 + 2 \alpha_0 \alpha_2 - 4 \alpha_4 \alpha_3 \bigg] + \frac{3 \mu^2}{\alpha_3} \bigg( \alpha_1 + 2 \alpha_2 + \nonumber \\
&& \sqrt{(\alpha_1 +  2 \alpha_2)^2 - 40 \alpha_3 \alpha_5} \bigg)  \frac{1}{\cosh^2(\xi)}
\end{eqnarray}
where
\begin{eqnarray*}
\xi &=& \mu x -\frac{\mu}{4 \alpha_3 \big( \sqrt{(\alpha_1 + 2 \alpha_2)^2 - 40 \alpha_3 \alpha_5} \alpha_1 + \alpha_1^2 + 2 \alpha_1 \alpha_2 + 2 \alpha_2^2 - 20 \alpha_3 \alpha_5 \big)} \bigg[ \nonumber \\
&& \big( \alpha_0^2 \alpha_1 + 4 \alpha_1 \alpha_3 + 16 \mu^4 \alpha_1^3 + 32 \mu^4 \alpha_1^2 \alpha_2 + 16 \mu^4 \alpha_1 \alpha_2^2 - \nonumber \\
&& 256 \mu^4 \alpha_1 \alpha_3 \alpha_5 \big) \sqrt{(\alpha_1  + 2 \alpha_2)^2 - 40 \alpha_3 \alpha_5} +64 \mu^4 \alpha_1^3 \alpha_2 + 80 \mu^4 \alpha_1^2 \alpha_2^2 - \nonumber \\
&& 576 \mu^4 \alpha_1^2 \alpha_3  \alpha_5 + 32 \mu^4 \alpha_1 \alpha_2^3 - 512 \mu^4 \alpha_1 \alpha_2 \alpha_3 \alpha_5 - 192 \mu^4 \alpha_2^2 \alpha_3 \alpha_5 + 1920 \mu^4 \alpha_3^2 \alpha_5^2 -\nonumber \\
&& \alpha_0^2 \alpha_1^2 - 2 \alpha_0^2 \alpha_1 \alpha_2 + 20 \alpha_0^2 \alpha_3 \alpha_5 - 8 \alpha_0 \alpha_2 \alpha_3 \alpha_4 + 8 \alpha_3^2 \alpha_4^2 + 4 \alpha_1^2 \alpha_3 + 8 \alpha_1 \alpha_2 \alpha_3 + \nonumber \\
&& 8 \alpha_2^2 \alpha_3 - 80 \alpha_3^2 \alpha_5\bigg]t
\end{eqnarray*}
\subsection{Case $m=4$}
In this case $q=2$. Then
\begin{equation}\label{d18}
u(\xi) = b_0 + b_1 g(\xi) + b_2 g^2(\xi); \ \ \xi = \mu x + \nu t
\end{equation}
and 
\begin{equation}\label{d19}
g_{(1)}^2 = a_0 + a_1 g + a_2 g^2 + a_3 g^3 + a_4 g^4
\end{equation}
In addition we have to solve the system of nonlinear algebraic equations for the
parameters of the solution that can be obtained from the equation $W_1(g) =0$
from Eqs.(\ref{d2}).
\par 
The general solution of Eq..(\ref{d19}) is given by the special function
$V_{a_0,a_1,a_2,a_3,a_4}(\xi; 1,2,4)$. For several particular cases this
function can be reduced to the elliptic functions of Jacobi. Examples are
\begin{eqnarray}\label{d20}
V_{1,0,-(1+k^2),0,k^2}(\xi;1,2,4) &=& \textrm{sn}(\xi;k); \nonumber \\
V_{1-k^2,0,(2k^2 -1),0,-k^2}(\xi;1,2,4) &=& \textrm{cn}(\xi;k); \nonumber \\
V_{-(1-k^2),0,(2-k^2),0,-1}(\xi;1,2,4) &=& \textrm{dn}(\xi;k).
\end{eqnarray}
\par 
For the general case the system of nonlinear algebraic equations becomes (\ref{d21x}) from the Appendix.
This system of equations possesses exact solutions for $a_0,a_1,a_2,a_3,a_4$ 
but they are very long and we shall not reproduce them here. Instead of this
we shall give a solution for an illustrative particular case.
From the first of equations (\ref{d21x}) one obtains a solution as follows
\begin{equation}\label{d22}
a_4 = \frac{b_2 [-\alpha_1 - 2 \alpha_2 + \sqrt{(\alpha_1 + 2 \alpha_2)^2 - 40 \alpha_3 \alpha_5}]}{120 \alpha_5 \mu^2}
\end{equation}
Let us consider the particular case 
\begin{equation}\label{d23x}
\alpha_5 = \frac{(\alpha_1+2 \alpha_2)^2 }{40 \alpha_3}
\end{equation}
For this particular case a solution of the system of equations
(\ref{d21x}) is as follows
\begin{eqnarray}\label{d24}
a_4 &=& - \frac{\alpha_3 b_2}{3 \mu^2 (\alpha_1^2 + 2 \alpha_2)}\nonumber \\
a_3 &=& - \frac{2 \alpha_3 b_1}{3 \mu^2 (\alpha_1 + 2 \alpha_2)}\nonumber \\
a_2 &=& - \frac{4 \alpha_1 \alpha_3 b_0 b_2 + \alpha_1 \alpha_3 b_1^2 + 2 \alpha_0 \alpha_1 b_2 + 4 \alpha_0 \alpha_2 b_2 - 8 \alpha_3 \alpha_4 b_2}{4 \mu^2 b_2 \alpha_1(\alpha_1 + 2 \alpha_2)}\nonumber \\
a_1 &=& - \frac{b_1(12 \alpha_1 \alpha_3 b_0 b_2-\alpha_1 \alpha_3 b_1^2 + 6 \alpha_0 \alpha_1 b_2 + 12 \alpha_0 \alpha_2 b_2 - 24 \alpha_3 \alpha_4 b_2)}{12 \mu^2 \alpha_1 b_2^2 (\alpha_1 + \alpha_2)}\nonumber \\
a_0 &=& - \frac{1}{48 \mu^3 \alpha_1^2 \alpha_2 b_2^3(\alpha_1 + 2 \alpha_2)(7 \alpha_1 - 6 \alpha_2)} \bigg( 336 \mu \alpha_1^3 \alpha_3^2 b_0^2 b_2^2 - 84 \mu \alpha_1^3 \alpha_3^2 b_0 b_1^2 b_2 + \nonumber \\
&& 7 \mu \alpha_1^3 \alpha_3^2 b_1^4 - 288 \mu \alpha_1^2 \alpha_2 \alpha_3^2 b_0^2 b_2^2 + 72 \mu \alpha_1^2 \alpha_2 \alpha_3^2 b_0 b_1^2 b_2 - 6 \mu \alpha_1^2 \alpha_2 \alpha_3^2 b_1^4 + \nonumber \\
&& 336 \mu \alpha_0 \alpha_1^3 \alpha_3 b_0 b_2^2 - 42 \mu \alpha_0 \alpha_1^3 \alpha_3 b_1^2 b_2 + 384 \mu \alpha_0 \alpha_1^2 \alpha_2 \alpha_3 b_0 b_2^2 - 48 \mu \alpha_0 \alpha_1^2 \alpha_2 \alpha_3 b_1^2 b_2 - \nonumber \\
&& 576 \mu \alpha_0 \alpha_1 \alpha_2^2 \alpha_3 b_0 b_2^2 + 72 \mu \alpha_0 \alpha_1 \alpha_2^2 \alpha_3 b_1^2 b_2 - 1344 \mu \alpha_1^2 \alpha_3^2 \alpha_4 b_0 b_2^2 + 168 \mu \alpha_1^2 \alpha_3^2 \alpha_4 b_1^2 b_2 + \nonumber \\
&& 1152 \mu \alpha_1 \alpha_2 \alpha_3^2 \alpha_4 b_0 b_2^2 - 144 \mu \alpha_1 \alpha_2 \alpha_3^2 \alpha_4 b_1^2 b_2 + 24 \mu \alpha_0^2 \alpha_1^3 b_2^2 + 144 \mu \alpha_0^2 \alpha_1^2 \alpha_2 b_2^2 + \nonumber \\
&& 288 \mu \alpha_0^2 \alpha_1 \alpha_2^2 b_2^2 + 192 \mu \alpha_0^2 \alpha_2^3 b_2^2 - 672 \mu \alpha_0 \alpha_1^2 \alpha_3 \alpha_4 b_2^2 - 1728 \mu \alpha_0 \alpha_1 \alpha_2 \alpha_3 \alpha_4 b_2^2 - \nonumber \\
&& 768 \mu \alpha_0 \alpha_2^2 \alpha_3 \alpha_4 b_2^2 + 2304 \mu \alpha_1 \alpha_3^2 \alpha_4^2 b_2^2 + 768 \mu \alpha_2 \alpha_3^2 \alpha_4^2 b_2^2 + 240 \mu \alpha_1^3 \alpha_3 b_2^2 + \nonumber \\
&& 480 \mu \alpha_1^2 \alpha_2 \alpha_3 b_2^2 + 240 \nu \alpha_1^3 \alpha_3 b_2^2 + 480 \nu \alpha_1^2 \alpha_2 \alpha_3 b_2^2\bigg)
\end{eqnarray}
and the solution of Eq.(\ref{d1}) for the discussed particular case (\ref{d23x}) is
\begin{eqnarray}\label{d23}
u(\xi) = b_0 + b_1 V_{a_0,a_1,a_2,a_3,a_4}(\xi;1,2,4) + b_2 V_{a_0,a_1,a_2,a_3,a_4}^2(\xi;1,2,4)
\end{eqnarray}
\par
Let us now consider the particular case when $V_{a_0,a_1,a_2,a_3,a_4}(\xi;1,2,4)$ is reduced to the Jacobi elliptic function $\textrm{sn}(\xi;k)$.
In this case we have to set $a_0=1; a_1=0; a_2=-(1+k^2); a_3 = 0; a_4 = k^2$
in Eqs.(\ref{d21x}). The full solution of the system of algebraic equation is again very
long. In order to illustrate it we discuss the particular case given by Eq.(\ref{d23x}).
A solution of the obtained system is
\begin{eqnarray*}
b_0 &=&  \frac{1}{2 \mu \alpha_1 \alpha_3 (7 \alpha_1 - 6 \alpha_2)} \bigg[
21 \mu^3 \alpha_1^3 + 24 \mu^3 \alpha_1^2 \alpha_2 - 36 \mu^3 \alpha_1 \alpha_2^2 - 7 \mu \alpha_0 \alpha_1^2 - \nonumber \\
&& 8 \mu \alpha_0 \alpha_1 \alpha_2 + 12 \mu \alpha_0 \alpha_2^2 + 28 \alpha_4 \mu \alpha_3 \alpha_1 - 24 \mu \alpha_2 \alpha_3 \alpha_4+ \bigg(-(7 \alpha_1^2 + 8 \alpha_1 \alpha_2 - 12 \alpha_2^2) \mu \times \nonumber \\
&& (21 \mu^5 \alpha_1^4 + 24 \mu^5 \alpha_1^3 \alpha_2 - 36 \mu^5 \alpha_1^2 \alpha_2^2 - 5 \mu \alpha_0^2 \alpha_1^2 + 20 \mu \alpha_0^2 \alpha_2^2 - 80 \mu \alpha_0 \alpha_2 \alpha_3 \alpha_4 + 80 \mu \alpha_3^2 \alpha_4^2 + \nonumber \\
&& 20 \mu \alpha_1^2 \alpha_3 + 20 \nu \alpha_1^2 \alpha_3)\bigg)^{1/2}
\bigg]
\end{eqnarray*}
\begin{eqnarray}\label{d25}
b_1&=& 0  \nonumber \\
b_2&=& -\frac{3}{2 \mu \alpha_1 \alpha_3 (7 \alpha_1 - 6 \alpha_2)}
\bigg[\mu^3 \alpha_1(7 \alpha_1^2 + 8 \alpha_1 \alpha_2 -12 \alpha_2^2)  + 
\bigg(-(7 \alpha_1^2 + 8 \alpha_1 \alpha_2 - 12 \alpha_2^2) \mu \times \nonumber \\
&& (21 \mu^5 \alpha_1^4 + 24 \mu^5 \alpha_1^3 \alpha_2 - 36 \mu^5 \alpha_1^2 \alpha_2^2 - 5 \mu \alpha_0^2 \alpha_1^2 + 20 \mu \alpha_0^2 \alpha_2^2 - 80 \mu \alpha_0 \alpha_2 \alpha_3 \alpha_4 + 80 \mu \alpha_3^2 \alpha_4^2 + \nonumber \\
&& 20 \mu \alpha_1^2 \alpha_3 + 20 \nu \alpha_1^2 \alpha_3)\bigg)^{1/2} \bigg]
\nonumber \\
k &=& \frac{1}{\sqrt{2} \mu \alpha_1 (\alpha_1 + 2 \alpha_2)} \bigg[ \frac{1}{\mu(7 \alpha_1 - 6 \alpha_2)} \bigg( \alpha_1 (\alpha_1 + 2 \alpha_2)(7 \mu^3 \alpha_1^3 + 8 \mu^3 \alpha_1^2 \alpha_2 - 12 \mu^3 \alpha_1 \alpha_2^2 + \nonumber \\
&& (-(7 \alpha_1^2 + 8 \alpha_1 \alpha_2 - 12 \alpha_2^2) \mu (21 \mu^5 \alpha_1^4 + 24 \mu^5 \alpha_1^3 \alpha_2 - 36 \mu^5 \alpha_1^2 \alpha_2^2 - 5 \mu \alpha_0^2 \alpha_1^2 + 20 \mu \alpha_0^2 \alpha_2^2 - \nonumber \\
&& 80 \mu \alpha_0 \alpha_2 \alpha_3 \alpha_4 + 80 \mu \alpha_3^2 \alpha_4^2 +20 \mu \alpha_1^2 \alpha_3 + 20 \nu \alpha_1^2 \alpha_3))^{1/2}) \bigg) \bigg]^{1/2}
\end{eqnarray}
and the solution of Eq.(\ref{d1}) is
\begin{eqnarray}\label{d26}
u(\xi) &=& \frac{1}{2 \mu \alpha_1 \alpha_3 (7 \alpha_1 - 6 \alpha_2)} \bigg[
21 \mu^3 \alpha_1^3 + 24 \mu^3 \alpha_1^2 \alpha_2 - 36 \mu^3 \alpha_1 \alpha_2^2 - 7 \mu \alpha_0 \alpha_1^2 - \nonumber \\
&& 8 \mu \alpha_0 \alpha_1 \alpha_2 + 12 \mu \alpha_0 \alpha_2^2 + 28 \alpha_4 \mu \alpha_3 \alpha_1 - 24 \mu \alpha_2 \alpha_3 \alpha_4+ \bigg(-(7 \alpha_1^2 + 8 \alpha_1 \alpha_2 - 12 \alpha_2^2) \mu \times \nonumber \\
&& (21 \mu^5 \alpha_1^4 + 24 \mu^5 \alpha_1^3 \alpha_2 - 36 \mu^5 \alpha_1^2 \alpha_2^2 - 5 \mu \alpha_0^2 \alpha_1^2 + 20 \mu \alpha_0^2 \alpha_2^2 - 80 \mu \alpha_0 \alpha_2 \alpha_3 \alpha_4 + 80 \mu \alpha_3^2 \alpha_4^2 + \nonumber \\
&& 20 \mu \alpha_1^2 \alpha_3 + 20 \nu \alpha_1^2 \alpha_3)\bigg)^{1/2} -
\frac{3}{2 \mu \alpha_1 \alpha_3 (7 \alpha_1 - 6 \alpha_2)}
\bigg[\mu^3 \alpha_1(7 \alpha_1^2 + 8 \alpha_1 \alpha_2 -12 \alpha_2^2)  + \nonumber \\
&&\bigg(-(7 \alpha_1^2 + 8 \alpha_1 \alpha_2 - 12 \alpha_2^2) \mu \times  (21 \mu^5 \alpha_1^4 + 24 \mu^5 \alpha_1^3 \alpha_2 - 36 \mu^5 \alpha_1^2 \alpha_2^2 - 5 \mu \alpha_0^2 \alpha_1^2 + \nonumber \\
&& 20 \mu \alpha_0^2 \alpha_2^2 - 80 \mu \alpha_0 \alpha_2 \alpha_3 \alpha_4 + 80 \mu \alpha_3^2 \alpha_4^2 + 
20 \mu \alpha_1^2 \alpha_3 + 20 \nu \alpha_1^2 \alpha_3)\bigg)^{1/2} \bigg] \times
\nonumber\\
&& \textrm{sn}^2 \bigg \{\xi; \frac{1}{\sqrt{2} \mu \alpha_1 (\alpha_1 + 2 \alpha_2)} \bigg[ \frac{1}{\mu(7 \alpha_1 - 6 \alpha_2)} \bigg( \alpha_1 (\alpha_1 + 2 \alpha_2)(7 \mu^3 \alpha_1^3 + 8 \mu^3 \alpha_1^2 \alpha_2 - \nonumber \\
&& 12 \mu^3 \alpha_1 \alpha_2^2 +  (-(7 \alpha_1^2 + 8 \alpha_1 \alpha_2 - 12 \alpha_2^2) \mu (21 \mu^5 \alpha_1^4 + 24 \mu^5 \alpha_1^3 \alpha_2 - 36 \mu^5 \alpha_1^2 \alpha_2^2 -
\nonumber \\
&& 5 \mu \alpha_0^2 \alpha_1^2 + 20 \mu \alpha_0^2 \alpha_2^2 - 80 \mu \alpha_0 \alpha_2 \alpha_3 \alpha_4 + 80 \mu \alpha_3^2 \alpha_4^2 +20 \mu \alpha_1^2 \alpha_3 + 20 \nu \alpha_1^2 \alpha_3))^{1/2}) \bigg) \bigg]^{1/2} \bigg\} \nonumber \\
&& \xi = \mu x + \nu t.
\end{eqnarray}
\par
Let us now consider the particular case when the simplest equation is the  
Riccati equation
\begin{equation}\label{d27}
g_{(1)} = c_0 + c_1g + c_2g^2
\end{equation}
In this case
\begin{equation}\label{d28}
a_0 = c_0^2; \ a_1 = 2 c_0c_1; \ a_2= c_1^2 + 2 c_0 c_2; \ a_3 = 2 c_1c_2; \
a_4 = c_2^2
\end{equation}
The full solution of the system of algebraic equation is again very long.
We illustrate it for the particular case given by Eq.(\ref{d23x}).
A solution of the system of algebraic equations is
\begin{eqnarray} \label{d28}
b_0 &=&  -\frac{8 \mu^2 \alpha_1^2 c_0 c_2 + \mu^2 \alpha_1^2 c_1^2 + 16 \mu^2 \alpha_1 \alpha_2 c_0 c_2 + 2 \mu^2 \alpha_1 \alpha_2 c_1^2 + 2 \alpha_0 \alpha_1 + 4 \alpha_0 \alpha_2 - 8 \alpha_3 \alpha_4}{4 \alpha_1 \alpha_3}\nonumber \\
b_1 &=& - \frac{3 \mu^2 c_1 c_2 (\alpha_1 +2 \alpha_2)}{\alpha_3}\nonumber \\
b_2 &=& - \frac{3 \mu^2 c_2^2 (\alpha_1 + 2 \alpha_2)}{\alpha_3} \nonumber \\
\nu &=& - \frac{\mu}{80 \alpha_1^2 \alpha_3} \bigg( 12 \mu^4 \alpha_1^4 c_0^2 c_2^2 - 56 \mu^4 \alpha_1^4 c_0 c_1^2 c_2 + 7 \mu^4 \alpha_1^4 c_1^4 + 128 \mu^4 \alpha_1^3 \alpha_2 c_0^2 c_2^2 -\nonumber \\
&& 64 \mu^4 \alpha_1^3 \alpha_2 c_0 c_1^2 c_2 + 8 \mu^4 \alpha_1^3 \alpha_2 c_1^4 - 192 \mu^4 \alpha_1^2 \alpha_2^2 c_0^2 c_2^2 + 96 \mu^4 \alpha_1^2 \alpha_2^2 c_0 c_1^2 c_2 - 
\nonumber \\
&& 12 \mu^4 \alpha_1^2 \alpha_2^2 c_1^4 - 20 \alpha_0^2 \alpha_1^2 + 80 \alpha_0^2 \alpha_2^2 - 320 \alpha_0 \alpha_2 \alpha_3 \alpha_4 + 320 \alpha_3^2 \alpha_4^2 + 80 \alpha_1^2 \alpha_3\bigg) 
\end{eqnarray}
From the solutions of the Riccati equation (\ref{d26}) here we shall discuss
the solution ($c$ is a constant of integration)
\begin{equation}\label{d29}
g(\xi) = \frac{1}{2c_2} \bigg\{ \tanh \bigg[ -\frac{(c_1^2 - 4 c_0 c_2)^{1/2}}{2}(\xi +c) \bigg]  - c_1\bigg \}
\end{equation}
valid when $c_1^2> 4 c_0 c_2 $ and $(2 c_2 g(\xi) + c_1)^2 < c_1^2 - 4 c_0 c_2$.
Then the solution of Eq.(\ref{d1}) is
\begin{eqnarray}\label{d30}
u(\xi) &=& 
-\frac{8 \mu^2 \alpha_1^2 c_0 c_2 + \mu^2 \alpha_1^2 c_1^2 + 16 \mu^2 \alpha_1 \alpha_2 c_0 c_2 + 2 \mu^2 \alpha_1 \alpha_2 c_1^2 + 2 \alpha_0 \alpha_1 + 4 \alpha_0 \alpha_2 - 8 \alpha_3 \alpha_4}{4 \alpha_1 \alpha_3}\nonumber \\
&& - \frac{3 \mu^2 c_1 (\alpha_1 +2 \alpha_2)}{2 \alpha_3}\bigg\{ \tanh \bigg[ -\frac{(c_1^2 - 4 c_0 c_2)^{1/2}}{2}(\xi +c) \bigg]  - c_1\bigg \} - \nonumber \\
&& - \frac{3 \mu^2  (\alpha_1 + 2 \alpha_2)}{4 \alpha_3} \bigg\{ \tanh \bigg[ -\frac{(c_1^2 - 4 c_0 c_2)^{1/2}}{2}(\xi +c) \bigg]  - c_1\bigg \}^2
\end{eqnarray}
which is a kink.
\subsection{Case $m = 5$}
In this case $q=3$. Then
\begin{equation}\label{d31}
u(\xi) = b_0 + b_1 g(\xi) + b_2 g^2(\xi) + b_3 g^3(\xi); \ \ \xi = \mu x + \nu t
\end{equation}
and 
\begin{equation}\label{d32}
g_{(1)}^2 = a_0 + a_1 g + a_2 g^2 + a_3 g^3 + a_4 g^4 + a_5 g^5
\end{equation}
In addition we have to solve the system of nonlinear algebraic equations for the
parameters of the solution that can be obtained from the equation $W_1(g) =0$
from Eqs.(\ref{d2}).
\par 
The general solution of Eq. (\ref{d32}) is given by the function
$V_{a_0,a_1,a_2,a_3,a_4,a_5}(\xi;1,2,5)$.
\par
The system of algebraic equations obtained from Eq.(\ref{d1}) is (\ref{d35x}).
The full solution of this system is quite long. We shall illustrate the solution 
for the particular case 
\begin{equation}\label{d36}
\alpha_2 = \alpha_3 = 1; \ \alpha_5 = \frac{(\alpha_1 + 2 \alpha_2)^2}{40 \alpha_3} = 
\frac{(\alpha_1 +2)^2}{40}
\end{equation}
The solution for this case is
\begin{eqnarray*}\label{d37a}
a_5 &=& - \frac{4b_3}{81 \mu^2} \nonumber \\
a_4 &=& - \frac{20 b_2}{243 \mu^2} \nonumber \\
a_3 &=& - \frac{4(b_2^2 + 7 b_1 b_3)}{243 \mu^2 b_3} \nonumber \\
a_2 &=&  - \frac{2(162 \alpha_1 b_0 b_3^2 + 72 \alpha_1 b_1 b_2 b_3 - 10 \alpha_1 b_2^3 + 243 \alpha_0 b_3^2 - 324 \alpha_4 b_3^2)}{2187 \mu^2 b_3^2 \alpha_1} \nonumber \\
\end{eqnarray*}
\begin{eqnarray} \label{d37}
a_1 &=& - \frac{4}{6561 \mu^2 \alpha_1 b_3^3} \bigg( 162 \alpha_1 b_0 b_2 b_3^2 + 108 \alpha_1 b_1^2 b_3^2 - 63 \alpha_1 b_1 b_2^2 b_3 + 8 \alpha_1 b_2^4 + \nonumber \\
&& 243 \alpha_0 b_2 b_3^2 - 324 \alpha_4 b_2 b_3 ^2 \bigg) \nonumber \\
a_0 &=& - \frac{2}{6561 \mu^2 \alpha_1 b_3^4 } \bigg( 648 \alpha_1 b_0 b_1 b_3^3 - 162 \alpha_1 b_0 b_2^2 b_3^2 - 144 \alpha_1 b_1 ^2 b_2 b_3^2 + 68 \alpha_1 b_1 b_2^3 b_3 - 8 \alpha_1 b_2^5 + \nonumber \\
&& 972 \alpha_0 b_1 b_3^3 - 243 \alpha_0 b_2^2 b_3^2 - 1296 \alpha_4 b_1 b_3^3 + 324 \alpha_4 b_2^2 b_3^2\bigg) \nonumber \\
\mu &=& -\frac{-65610 \nu \alpha_1 b_3^4}{T^*} \nonumber \\
T^* &=& 4374 \alpha_1^2 b_0^2 b_3^4 - 2916 \alpha_1^2 b_0 b_1 b_2 b_3^3 + 648 \alpha_1^2 b_0 b_2^3 b_3^2 - 216 \alpha_1^2 b_1^3 b_3^3 + 702 \alpha_1^2 b_1^2 b_2^2 b_3^2 - \nonumber \\
&& 288 \alpha_1^2 b_1 b_2^4 b_3 + 32 \alpha_1^2 b_2^6 + 13122 \alpha_0 \alpha_1 b_0 b_3^4 - 4374 \alpha_0 \alpha_1 b_1 b_2 b_3^3 + 972 \alpha_0 \alpha_1 b_2^3 b_3^2 - \nonumber \\
&& 17496 \alpha_1 \alpha_4 b_0 b_3^4 + 5832 \alpha_1 \alpha_4 b_1 b_2 b_3^3 - 1296 \alpha_1 \alpha_4 b_2^3 b_3^2 + 59049 \alpha_0^2 b_3^4 - \nonumber \\
&& 288684 \alpha_0 \alpha_4 b_3^4 + 279936 \alpha_4^2 b_3^4 + 65610 \alpha_1 b_3^4
\end{eqnarray}
and the solution of Eq.(\ref{d1}) (for the particular case given by Eqs.(\ref{d36}))
is
\begin{eqnarray}\label{d38}
u(\xi) &=& b_0 + b_1 V_{a_0,a_1,a_2,a_3,a_4,a_5}(\xi; 1,2,5) + b_2 V^2_{a_0,a_1,a_2,a_3,a_4,a_5}(\xi; 1,2,5) + \nonumber \\
&& b_3  V^3_{a_0,a_1,a_2,a_3,a_4,a_5}(\xi; 1,2,5) 
\end{eqnarray}
where 
\begin{eqnarray}\label{d39}
\xi =\frac{65610 \nu \alpha_1 b_3^4}{T^*} x + \nu t 
\end{eqnarray}
\subsection{Case $m=6$}
In this case $q=4$. We note that with increasing $m$ (and $q$) the number of equations
in the nonlinear algebraic system we have to solve becomes large very fast. When the 
number of equations become larger that the number of parameters of the solution then
the system could not have any nontrivial solution.
\par 
For the case $m=6$
\begin{equation}\label{d40}
u(\xi) = b_0 + b_1 g(\xi) + b_2 g^2(\xi) + b_3 g^3(\xi) + b_4 g^4(\xi); \ \ \xi = \mu x + \nu t
\end{equation}
and 
\begin{equation}\label{d41}
g_{(1)}^2 = a_0 + a_1 g + a_2 g^2 + a_3 g^3 + a_4 g^4 + a_5 g^5 + a_6 g^6
\end{equation}
In addition we have to solve the system of nonlinear algebraic equations for the
parameters of the solution that can be obtained from the equation $W_1(g) =0$
from Eqs.(\ref{d2}).
The general solution of Eq.(\ref{d41}) is given by the function $V_{a_0,a_1,a_2,a_3,a_4,a_5,a_6}(\xi; 1,2,6)$.
\par
The system of nonlinear algebraic equations connected to Eq.(\ref{d1})
is very large. It still has a solution. In order to illustrate this solution we 
shall consider a particular case of Eq.(\ref{d41}) 
namely the equation of Abel of first kind
\begin{equation}\label{d42}
g_{(1)} = c_0 + c_1 g + c_2 g^2 + c_3 g^3
\end{equation}
The corresponding system of nonlinear algebraic equations consists of 12 equations.
The solution is very long and we shall illustrate it for the particular case
when $\alpha_5 = \frac{(\alpha_1 + 2 \alpha_2)^2}{40 \alpha_3}$.  For this case
one solution of the system of algebraic equations is
\begin{eqnarray}\label{d43}
b_4&=& -\frac{4 \mu^2 c_2^4 (\alpha_1 + 2 \alpha_2)}{3 \alpha_3 c_1^2} \nonumber \\
b_3&=& - \frac{16 \mu^2 c_2^3  (\alpha_1 + 2 \alpha_2}{3 \alpha_3 c_1} \nonumber \\
b_2&=& - \frac{8 \mu^2 c_2^2 (\alpha_1 + 2 \alpha_2)}{\alpha_3} \nonumber \\
b_1 &=& - \frac{16 \mu^2 c_2 (3 \alpha_1 c_0 c_2 + 4 \alpha_1 c_1^2 + 6 \alpha_2 c_0 c_2 + 8 \alpha_2 c_1^2)}{15 \alpha_3 c_1}\nonumber \\
b_0 &=& - \frac{1}{30 \alpha_1 \alpha_3} \bigg( 96 \mu^2 \alpha_1^2 c_0 c_2 + 8 \mu^2 \alpha_1^2 c_1^2 + 192 \mu^2 \alpha_1 \alpha_2 c_0 c_2 + 16 \mu^2 \alpha_1 \alpha_2 c_1^2 + \nonumber \\
&& 15 \alpha_0 \alpha_1 + 30 \alpha_0 \alpha_2 - 60 \alpha_3 \alpha_4 \bigg) \nonumber \\
c_3 &=& \frac{c_1^3}{27 c_0^2}\nonumber \\
c_2 &=& \frac{c_1^2}{3 c_0} \nonumber \\
\mu &=&  \frac{4 \nu \alpha_1^2 \alpha_3}{\alpha_0^2 \alpha_1^2 - 4 \alpha_0^2 \alpha_2^2 + 16 \alpha_0 \alpha_2 \alpha_3 \alpha_4 - 16 \alpha_3^2 \alpha_4^2 - 4 \alpha_1^2 \alpha_3}
\end{eqnarray}
The Abel equations becomes: 
\begin{equation}\label{d44}
g_{(1)}=c_0 + c_1 g + \frac{c_1^2}{3 c_0} g^2 + \frac{c_1^3}{27 c_0^2} g^3
\end{equation}
Its solution is:
\begin{equation}\label{d45}
g(\xi) = V_{c_0^2,2 c_0 c_1,\frac{5c_1^2}{3},\frac{20 c_1^3}{27 c_0},\frac{5 c_1^4}{27 c_0^2},\frac{2 c_1^5}{81 c_0^3},\frac{c_1^6}{729 c_0^4}}
\end{equation}
The solution of Eq. (\ref{d1}) is 
\begin{eqnarray}\label{d46}
u(\xi) &=& b_0 + b_1V_{c_0^2,2 c_0 c_1,\frac{5c_1^2}{3},\frac{20 c_1^3}{27 c_0},\frac{5 c_1^4}{27 c_0^2},\frac{2 c_1^5}{81 c_0^3},\frac{c_1^6}{729 c_0^4}} + b_2 V^2_{c_0^2,2 c_0 c_1,\frac{5c_1^2}{3},\frac{20 c_1^3}{27 c_0},\frac{5 c_1^4}{27 c_0^2},\frac{2 c_1^5}{81 c_0^3},\frac{c_1^6}{729 c_0^4}} + \nonumber \\
&& b_3 V^3_{c_0^2,2 c_0 c_1,\frac{5c_1^2}{3},\frac{20 c_1^3}{27 c_0},\frac{5 c_1^4}{27 c_0^2},\frac{2 c_1^5}{81 c_0^3},\frac{c_1^6}{729 c_0^4}}
+b_4 V^4_{c_0^2,2 c_0 c_1,\frac{5c_1^2}{3},\frac{20 c_1^3}{27 c_0},\frac{5 c_1^4}{27 c_0^2},\frac{2 c_1^5}{81 c_0^3},\frac{c_1^6}{729 c_0^4}}
\end{eqnarray}
where
\begin{eqnarray*}
\xi &=&  \frac{4 \nu \alpha_1^2 \alpha_3}{\alpha_0^2 \alpha_1^2 - 4 \alpha_0^2 \alpha_2^2 + 16 \alpha_0 \alpha_2 \alpha_3 \alpha_4 - 16 \alpha_3^2 \alpha_4^2 - 4 \alpha_1^2 \alpha_3} x + \nu t
\end{eqnarray*}
Let now $c_0 = \frac{c_2}{3 c_3} \left( c_1 - \frac{2 c_2^2}{9 c_3}\right)$ and $c_3<0$. Then the Abel equation
\begin{equation}\label{d47}
g_{(1)}= \frac{c_2}{3 c_3} \left( c_1 - \frac{2 c_2^2}{9 c_3} + c_1 g + c_2 g^2 + 
c_3 g^3 \right)
\end{equation}
has the solution \cite{kamke}
\begin{equation}\label{d48}
g(\xi) = \frac{\exp\left[\left( c_1 - \frac{c_2^2}{3 c_3} \right) \xi \right]}{\sqrt{
C - 2 c_3\exp\left[2\left( c_1 - \frac{c_2^2}{3 c_3} \right) \xi \right] }} - \frac{c_2}{3 c_3}
\end{equation}
where $C$ is the constant of integration. Let us assume that $C>0$. The use of the Abel equation (\ref{d47}) leads to a following solution of the system of nonlinear algebraic equations (for the particular case $\alpha_5 = \frac{(\alpha_1 + 2 \alpha_2)^2}{40 \alpha_3}$)
\begin{eqnarray}\label{d49}
b_4 &=& - \frac{12 \mu^2 c_3^2 (\alpha_1 + 2 \alpha_2)}{\alpha_3} \nonumber \\
b_3 &=& -\frac{16 \mu^2 c_2 c_3 (\alpha_1 + 2\alpha_2)}{\alpha_3} \nonumber \\
b_2 &=& - \frac{4 \mu^2 (3 \alpha_1 c_1 c_3 + \alpha_1 c_2^2 + 6 \alpha_2 c_1 c_3 + 2 \alpha_2 c_2^2)}{\alpha_3} \nonumber \\
b_1 &=& - \frac{8 \mu^2 (9 \alpha_1 c_1 c_3 - \alpha_1 c_2^2 + 18 \alpha_2 c_1 c_3 - 2 \alpha_2 c_2^2) }{9 \alpha_3 c_3 } \nonumber \\
\nu &=& - \frac{\mu}{1620 \alpha_1^2 \alpha_3 c_3^4} \bigg[2268 \mu^4 \alpha_1^4 c_1^4 c_3^4 - 3024 \mu^4 \alpha_1^4 c_1^3 c_2^2 c_3^3 + 1512 \mu^4 \alpha_1^4 c_1^2 c_2^4 c_3^2 - \nonumber \\
&& 336 \mu^4 \alpha_1^4 c_1 c_2^6 c_3 + 28 \mu^4 \alpha_1^4 c_2^8 + 2592 \mu^4 \alpha_1^3 \alpha_2 c_1^4 c_3^4 - 3456 \mu^4 \alpha_1^3 \alpha_2 c_1^3 c_2^2 c_3^3 + \nonumber \\
&& 1728 \mu^4 \alpha[1]^3 \alpha_2 c_1^2 c_2^4 c_3^2 - 384 \mu^4 \alpha_1^3 \alpha_2 c_1 c_2^6 c_3 + 32 \mu^4 \alpha_1^3 \alpha_2 c_2^8 - 3888 \mu^4 \alpha_1^2 \alpha_2^2 c_1^4 c_3^4 + \nonumber \\
&& 5184 \mu^4 \alpha_1^2 \alpha_2^2 c_1^3 c_2^2 c_3^3 - 2592 \mu^4 \alpha_1^2 \alpha_2^2 c_1^2 c_2^4 c_3^2 + 576 \mu^4 \alpha_1^2 \alpha_2^2 c_1 c_2^6 c_3 - 48 \mu^4 \alpha_1^2 \alpha_2^2 c_2^8 - \nonumber \\
&& 405 \alpha_0^2 \alpha_1^2 c_3^4 + 1620 \alpha_0^2 \alpha_2^2 c_3^4 - 6480 \alpha_0 \alpha_2 \alpha_3 \alpha_4 c_3^4 + 6480 \alpha_3^2 \alpha_4^2 c_3^4 + 1620 \alpha_1^2 \alpha_3 c_3^4 \bigg] \nonumber \\
\end{eqnarray}
Hence the solution of Eq.(\ref{d1}) becomes
\begin{eqnarray}\label{d50}
u(\xi)&=& b_0  - \frac{8 \mu^2 (9 \alpha_1 c_1 c_3 - \alpha_1 c_2^2 + 18 \alpha_2 c_1 c_3 - 2 \alpha_2 c_2^2) }{9 \alpha_3 c_3 } \times \nonumber \\
&& \left(\frac{\exp\left[\left( c_1 - \frac{c_2^2}{3 c_3} \right) \xi \right]}{\sqrt{
C - 2 c_3\exp\left[2\left( c_1 - \frac{c_2^2}{3 c_3} \right) \xi \right] }} - \frac{c_2}{3 c_3} \right) - \nonumber \\
&& - \frac{4 \mu^2 (3 \alpha_1 c_1 c_3 + \alpha_1 c_2^2 + 6 \alpha_2 c_1 c_3 + 2 \alpha_2 c_2^2)}{\alpha_3} \times \nonumber \\
&& \left( \frac{\exp\left[ \left( c_1 - \frac{c_2^2}{3 c_3} \right) \xi \right]}{\sqrt{
C - 2 c_3\exp\left[2\left( c_1 - \frac{c_2^2}{3 c_3} \right) \xi \right] }} - \frac{c_2}{3 c_3} \right)^2 - \nonumber \\
&& -\frac{16 \mu^2 c_2 c_3 (\alpha_1 + 2\alpha_2)}{\alpha_3} \left(\frac{\exp\left[\left( c_1 - \frac{c_2^2}{3 c_3} \right) \xi \right]}{\sqrt{
C - 2 c_3\exp\left[2\left( c_1 - \frac{c_2^2}{3 c_3} \right) \xi \right] }} - \frac{c_2}{3 c_3}\right)^3 - \nonumber \\
&& - \frac{12 \mu^2 c_3^2 (\alpha_1 + 2 \alpha_2)}{\alpha_3}  \left( \frac{\exp\left[\left( c_1 - \frac{c_2^2}{3 c_3} \right) \xi \right]}{\sqrt{
C - 2 c_3\exp\left[2\left( c_1 - \frac{c_2^2}{3 c_3} \right) \xi \right] }} - \frac{c_2}{3 c_3} \right)^4
\nonumber \\
\end{eqnarray}
that is a solution of kink's kind as the solution (\ref{d30}).
\section{Concluding remarks}
In this article we have discussed a version of the method of simplest equation applicable
to a class of nonlinear partial differential equations that are much used as model equations
in the area of natural sciences. Eq.(\ref{a9}) was used as simplest equation and we have
described a methodology based on the concept for balance equations. This methodology reduces the 
studied nonlinear partial differential equations to systems of nonlinear algebraic equations. Any nontrivial
solution of the obtained system of algebraic equations  leads to an exact solution of the 
corresponding nonlinear partial differential equation. Discussed examples have shown the 
effectivity of the methodology of this variant of the method of simplest equation.

\begin{appendix}
\section{Systems of algebraic equations connected to Eq.(\ref{d1})}
\subsection{Case $m=3$}
\begin{eqnarray}\label{d13x}
0&=& \frac{3}{2} \alpha_1 \mu^2 a_3 b_1 + 3 \alpha_2 \mu^2 b_1 a_3 
+ \alpha_3  b_1^2 + \frac{45}{2} \alpha_5 \mu^4 a_3^2\nonumber \\
0&=& 15 \mu^4 a_2 a_3 \alpha_5 + \mu^2 a_2 \alpha_1 b_1 + \mu^2 a_2 \alpha_2 b_1 + 3 \mu^2 a_3 \alpha_2 b_0 + 3 \mu^2 a_3 \alpha_4  + \nonumber \\
&& 2 \alpha_3 b_0 b_1 + 
\alpha_0 b_1\nonumber \\
0&=&\nu + \mu + \alpha_0 \mu b_0 + \frac{1}{2} \alpha_1 \mu^3 a_1 b_1 + 
\alpha_2 \mu^3 b_0 a_2  + \alpha_3 \mu b_0^2  + \nonumber \\
&& \alpha_4 \mu^3 a_2 + \alpha_5 \mu^5 \bigg(\frac{9}{2} a_3  a_1 + a_2^2  \bigg)
\end{eqnarray}
\subsection{Case $m=4$}
\begin{eqnarray}\label{d21x}
0&=&360 \mu^4 a_4^2 \alpha_5  + 6 \mu^2 a_4 \alpha_1 b_2 + 12 \mu^2 a_4   \alpha_2 b_2  + \alpha_3 b_2^2 \nonumber \\
0&=&2 \alpha_1 \mu^2  b_2 (5 a_3 b_2 + 2 a_4 b_1) + 
6 \alpha_1 \mu^2 a_4 b_2 b_1 + 
24 \alpha_2 \mu^2   b_1 a_4 b_2 + \nonumber \\
&&\alpha_2 \mu^2 b_2 (15 a_3 b_2 + 6 a_4 b_1) + 
5 \alpha_3  b_1 b_2^2   + \alpha_5 \mu^4 (840 a_3 a_4 b_2 + 120 a_4^2 b_1) \nonumber \\
0&=&2 \alpha_0 b_2^2 + 2 \alpha_1 \mu^2 (4 a_2 b_2 + (3/2) a_3 b_1) b_2 +
\alpha_1 \mu^2 (5 a_3 b_2 + 2 a_4 b_1) b_1 + \nonumber \\
&&24 \alpha_2 \mu^2 b_0 a _4 b_2 + \alpha_2 \mu^2 b_1(15 a_3 b_2 + 6 a_4 b_1) + 
\alpha_2 \mu^2 b_2 (8 a_2 b_2 +3 a_3 b_1) +  \nonumber \\
&& 2 \alpha_3(2 b_0 b_2 + b_1^2)b_2 +  2 \alpha_3 b_1^2 b_2 + 
24 \alpha_4 \mu^2 a_4 b_2 + \alpha_5 \mu^4 (480 a_2 a_4 b_2 + 
210 a_3^2 b_2 + \nonumber \\
&& 120 a_3 a_4 b_1)  \nonumber \\
0&=& 3 \alpha_0  b_1 b_2 + 2 \alpha_1 \mu^2 (3 a_1 b_2 + a_2 b_1) b_2 + 
\alpha_1 \mu^2 (4 a_2 b_2 + (3/2) a_3 b_1) b_1 + \nonumber \\
&& \alpha_2 \mu^2 b_0 (15 a_3 b_2 + 6 a_4 b_1) + \alpha_2 \mu^2 b_1 (8 a_2 b_2 
+ 3 a_3 b_1) + 
\alpha_2 \mu^2 b_2(3 a_1 b_2 + a_2 b_1) + \nonumber \\
&& 4 \alpha_3 b_0 b_1 b_2 + \alpha_3 b_1(2 b_0 b_2 + b_1^2) + \alpha_4 \mu^2 (15 a_3 b_2 + 6 a_4 b_1) + 
\alpha_5 \mu^4[(45/2)a_3^2 b_1 + \nonumber \\
&& 270 a_4 b_2 a_1 + 195 a_3 b_2 a_2 + 60 a_4 b_1 a_2]\nonumber \\
0&=& 2 (\nu + \mu) b_2 + 2 \alpha_0 \mu b_0 b_2 + 
\alpha_0 \mu b_1^2 + 2 \alpha_1 \mu^3 (2 a_0 b_2 + (1/2) a_1 b_1) b_2 + 
\nonumber \\
&& \alpha_1 \mu^3 (3 a_1 b_2 + a_2 b_1) b_1 + \alpha_2  \mu^3 b_0 (8 a_2 b_2 + 
3 a_3 b_1) + \alpha_2 \mu^3 b_1 (3 a_1 b_2 + a_2 b_1) + \nonumber \\
&& 2 \alpha_3 \mu b_0^2 b_2 + 2 \alpha_3 \mu b_0 b_1^2   + \alpha_4 \mu^3 (8 a_2 b_2 + 
3 a_3 b_1) + \alpha_5 \mu^5 (144 a_0 a_4 b_2 + \nonumber \\
&& 84 a_1 a_3b _2 + 30 a_1 a_4 b_1 + 32 a_2^2 b_2 + 15 a_2 a_3 b_1) \nonumber \\
0&=& (\nu + \mu) b_1 + \alpha_0 \mu b_0 b_1 + \alpha_1 \mu^3 (2 a_0 b_2 +
(1/2) a_1 b_1) b_1 + \alpha_2 \mu^3 b_0 (3 a_1 b_2 + a_2 b_1) +\nonumber \\
&&\alpha_3 \mu b_0^2 b_1 + \alpha_4 \mu^3 (3 a_1 b_2 + a_2 b_1) + 
\alpha_5 \mu^5 (a_2^2 b_1 + 12 a_4 b_1 a_0 + 15 a_2 b_2 a_1 + \nonumber \\
&& 30 a_3 b_2 a_0 + (9/2) a_3 b_1 a_1)
\end{eqnarray}
\subsection{Case $m=5$}
\begin{eqnarray*}\label{d34x}
0&=&(81/2) \alpha_1 \mu^3 a_5 b_3^2 + 81 \alpha_2 \mu^3 b_3^2 a_5 + 3 \alpha_3 \mu b_3^3 + (10935/2) \alpha_5 \mu^5 a_5^2 b_3 \nonumber \\
0&=&\alpha_1 \mu^3 (12 a_4 b_3 + 7 a_5 b_2) b_3 + 27 \alpha_1 \mu^3 a_5 b_3 b_2 + 81 \alpha_2 \mu^3 b_2 a_5 b_3 + \nonumber \\
&& \alpha_2 \mu^3 b_3 (60 a_4 b_3 + 35 a_5 b_2) + 8 \alpha_3 \mu b_2 b_3^2 + \alpha_5 \mu^5 (7656 a_4 a_5 b_3 + 1820 a_5^2 b_2) \nonumber \\
0&=&3 \alpha_1 \mu^3 ((21/2) a_3 b_3 + 6 a_4 b_2 + (5/2) a_5 b_1) b_3 + 2 \alpha_1 \mu^3 (12 a_4 b_3 + 7 a_5 b_2) b_2 + \nonumber \\
&& (27/2) \alpha_1 \mu^3 a_5 b_3 b_1 + 81 \alpha_2 \mu^3 b_1 a_5 b_3 + \alpha_2 \mu^3 b_2 (60 a_4 b_3 + 35 a_5 b_2) + \nonumber \\
&& \alpha_2 \mu^3 b_3 (42 a_3 b_3 + 24 a_4 b_2 + 10 a_5 b_1) + 3 \alpha_3 \mu (2 b_1 b_3 + b_2^2) b_3 + 4 \alpha_3 \mu b_2^2 b_3 + \nonumber \\
&& \alpha_3 \mu b_3^2 b_1 + \alpha_5 \mu^5 (2520 a_4^2 b_3 + 385 a_5^2 b_1 + (10605/2) a_5 b_3 a_3 + 2394 a_5 b_2 a_4) \nonumber \\
0&=& 3 \alpha_0 \mu b_3^2 + 3 \alpha_1 \mu^3 (9 a_2 b_3 + 5 a_3 b_2 + 2 a_4 b_1) b_3 + 2 \alpha_1 \mu^3 ((21/2) a_3 b_3 + 6 a_4 b_2 + \nonumber \\
&& (5/2) a_5 b_1) b_2 + \alpha_1 \mu^3 (12 a_4 b_3 + 7 a_5 b_2) b_1 + 81 \alpha_2 \mu^3 b_0 a_5 b_3 + \nonumber \\
&& \alpha_2 \mu^3 b_1 (60 a_4 b_3 + 35 a_5 b_2) + \alpha_2 \mu^3 b_2 (42 a_3 b_3 + 24 a_4 b_2 + 10 a_5 b_1) + \nonumber \\
&& \alpha_2 \mu^3 b_3 (27 a_2 b_3 + 15 a_3 b_2 + 6 a_4 b_1) + 3 \alpha_3 \mu (2 b_0 b_3 + 2 b_1 b_2) b_3 + 2 \alpha_3 \mu (2 b_1 b_3 + b_2^2) b_2 + \nonumber \\
&& 2 \alpha_3 \mu b_2 b_3 b_1 + 81 \alpha_4 \mu^3 a_5 b_3 + \alpha_5 \mu^5 (3645 a_2 a_5 b_3 + 3240 a_3 a_4 b_3 + \nonumber \\
&& 1560 a_3 a_5 b_2 + 720 a_4^2 b_2 + 462 a_4 a_5 b_1) \nonumber \\
0&=& 5\alpha_0 \mu b_2 b_3 + 3 \alpha_1 \mu^3 ((15/2) a_1 b_3 + 4 a_2 b_2 + (3/2) a_3 b_1) b_3 + 2 \alpha_1 \mu^3 (9 a_2 b_3 + \nonumber \\
&& 5 a_3 b_2 + 2 a_4 b_1) b_2 + \alpha_1 \mu^3 ((21/2) a_3 b_3 + 6 a_4 b_2 + (5/2) a_5 b_1) b_1 + \nonumber \\
&& \alpha_2 \mu^3 b_0 (60 a_4 b_3  + 35 a_5 b_2) + \alpha_2 \mu^3 b_1 (42 a_3 b_3 + 24 a_4 b_2 + 10 a_5 b_1) + \nonumber \\
&& \alpha_2 \mu^3 b_2 (27 a_2 b_3 + 15 a_3 b_2 + 6 a_4 b_1) + \alpha_2 \mu^3 b_3 (15 a_1 b_3 + 8 a_2 b_2 + 3 a_3 b_1) + \nonumber \\
&& 3 \alpha_3 \mu (2 b_0 b_2 + b_1^2) b_3 + 2 \alpha_3 \mu (2 b_0 b_3 + 2 b_1 b_2) b_2 + \alpha_3 \mu (2 b_1 b_3 + b_2^2) b_1 + \nonumber \\
&& \alpha_4 \mu^3 (60 a_4 b_3 + 35 a_5 b_2) + \alpha_5 \mu^5 (2490 a_5 b_3 a_1 + 2040 a_4 b_3 a_2 + 1015 a_5 b_2 a_2 + \nonumber \\
&& 945 a_3^2 b_3 + 840 a_4 b_2 a_3 + 120 a_4^2 b_1 + (555/2) a_5 b_1 a_3) \nonumber \\
0&=& 4 \alpha_0 \mu b_1 b_3 + 2 \alpha_0 \mu b_2^2 + 3 \alpha_1 \mu^3 (6 a_0 b_3 + 3 a_1 b_2 + a_2 b_1) b_3 + 2 \alpha_1 \mu^3 ((15/2) a_1 b_3 + \nonumber \\
&& 4 a_2 b_2 + (3/2) a_3 b_1) b_2 + \alpha_1 \mu^3 (9 a_2 b_3 + 5 a_3 b_2 + 2 a_4 b_1) b_1 + \alpha_2 \mu^3 b_0 (42 a_3 b_3 + \nonumber \\
&& 24 a_4 b_2 + 10 a_5 b_1) + \alpha_2 \mu^3 b_1 (27 a_2 b_3 + 15 a_3 b_2 + 6 a_4 b_1) + \alpha_2 \mu^3 b_2 (15 a_1 b_3 + 8 a_2 b_2 + \nonumber \\
&& 3 a_3 b_1) + \alpha_2 \mu^3 b_3 (6 a_0 b_3 + 3 a_1 b_2 + a_2 b_1) + 6 \alpha_3 \mu b_0 b_1 b_3 + 2 \alpha_3 \mu (2 b_0 b_2 + b_1^2) b_2 + \nonumber \\
&& \alpha_3 \mu (2 b_0 b_3 + 2 b_1 b_2) b_1 + \alpha_4 \mu^3 (42 a_3 b_3 + 24 a_4 b_2 + 10 a_5 b_1) + \alpha_5 \mu^5 (1680 a_0 a_5 b_3 + \nonumber \\
&& 1260 a_1 a_4 b_3 + 660 a_1 a_5 b_2 + 1050 a_2 a_3 b_3 + 480 a_2 a_4 b_2 + 170 a_2 a_5 b_1 + 210 a_3^2 b_2 + \nonumber \\
&& 120 a_3 a_4 b_1) \nonumber \\
\end{eqnarray*}
\begin{eqnarray}\label{d35x}
&=& (3 (\nu + \mu)) b_3 + 3 \alpha_0 \mu b_0 b_3 + 3 \alpha_0 \mu b_1 b_2 + 3 \alpha_1 \mu^3 (2 a_0 b_2 + (1/2) a_1 b_1 ) b_3 + \nonumber \\
&& 2 \alpha_1 \mu^3 (6 a_0 b _3 + 3 a_1 b_2 + a_2 b_1) b_2 + \alpha_1 \mu^3 ((15/2) a_1 b_3 + 4 a_2 b_2 + \nonumber \\
&& (3/2) a_3 b_1) b_1 + \alpha_2 \mu^3 b_0 (27 a_2 b_3 + 15 a_3 b_2 + 6 a_4 b_1) + \alpha_2 \mu^3 b_1 (15 a_1 b_3 + 8 a_2 b_2 + 3 a_3 b_1) + \nonumber \\
&& \alpha_2 \mu^3 b_2 (6 a_0 b_3 + 3 a_1 b_2 + a_2 b_1) + 3 \alpha_3 \mu b_0^2 b_3 + 4 \alpha_3 \mu b_0 b_1 b_2 + \alpha_3 \mu (2 b_0 b_2 + b_1^2) b_1 + \nonumber \\
&& \alpha_4 \mu^3 (27 a_2 b_3 + 15 a_3 b_2 + 6 a_4 b_1) + \alpha_5 \mu^5 (243 a_2^2 b_3 + (45/2) a_3^2 b_1 + 756 a_4 b_3 a_0 + \nonumber \\
&& 420 a_5 b_2 a_0 + (1107/2) a_3 b_3 a_1 + 270 a_4 b_2 a_1 + 105 a_5 b_1 a_1 + 195 a_3 b_2 a_2 + 60 a_4 b_1 a_2) \nonumber \\
0&=& (2 (\nu + \mu)) b_2 + 2 \alpha_0 \mu b_0 b_2+ \alpha_0 \mu b_1^2 + 2 \alpha_1 \mu^3 (2 a_0 b_2 + (1/2) a_1 b_1) b_2 + \nonumber \\
&& \alpha_1 \mu^3 (6 a_0 b_3 + 3 a_1 b_2 + a_2 b_1) b_1 + \alpha_2 \mu^3 b_0 (15 a_1 b_3   + 8 a_2 b_2 + 3 a_3 b_1) + \nonumber \\
&& \alpha_2 \mu^3 b_1 (6 a_0 b_3 + 3 a_1 b_2 + a_2 b_1) + 2 \alpha_3 \mu b_0^2 b_2 + 2 \alpha_3 \mu b_0 b_1^2 + \alpha_4 \mu^3 (15 a_1 b_3 + 8 a_2 b_2 + 3 a_3 b_1) + \nonumber \\
&& \alpha_5 \mu^5 (270 a_0 a_3 b_3 + 144 a_0 a_4 b_2 + 60 a_0 a_5 b_1 + 195 a_1 a_2 b_3 + 84 a_1 a_3 b_2 + 30 a_1 a_4 b_1 +\nonumber \\
&& 32 a_2^2 b_2 + 15 a_2 a_3 b_1) \nonumber \\
0&=& (\nu+\mu) b_1 + \alpha_0 \mu b_0 b_1 + \alpha_1 \mu^3 (2 a_0 b_2 + (1/2) a_1 b_1) b_1 + \alpha_2 \mu^3 b_0 (6 a_0 b_3 + 3 a_1 b_2 + a_2 b_1 ) + \nonumber \\
&& \alpha_3 \mu b_0^2 b_1 + \alpha_4 \mu^3 (6 a_0 b_3 + 3 a_1 b_2 + a_2 b_1) + \alpha_5 \mu^5 ((45/2) a_1^2 b_3 + a_2^2 b_1 + (9/2) a_3 b_1 a_1 + \nonumber \\
&& 15 a_2 b_2 a_1 + 12 a_4 b_1 a_0 + 60 a_2 b_3 a_0 + 30 a_3 b_2 a_0)
\nonumber \\
\end{eqnarray}

\end{appendix}

\end{document}